\documentclass[usenames]{lmcs} 


\usepackage{hyperref}
\theoremstyle{plain} 

\usepackage{bm}
\usepackage{color}
\usepackage{url}
\usepackage{float}
\usepackage{fancyvrb}
\usepackage{amsmath,amsfonts,amsxtra,amssymb}

\usepackage[caption=false]{subfig}
\usepackage{alltt}
\usepackage{listings}
\usepackage{graphicx}

\definecolor{lightblue}{rgb}{0, 0, 0}

\begin{document}

\title{Modelling and Verifying Neuronal Archetypes in Rocq}

\author[A.~Bahrami]{Abdorrahim Bahrami}[a]
\author[R.~Zucchini]{R\'ebecca Zucchini}[a]
\author[E.~De Maria]{Elisabetta De Maria}[b]
\author[A.~Felty]{Amy Felty}[a]

\address{School of Electrical Engineering and Computer Science,
  University of Ottawa, Canada}	
\email{abahrami@uottawa.ca, zucchinii.rebecca@gmail.com, afelty@uottawa.ca}  

\address{Universit\'{e} C\^{o}te d'Azur, CNRS, I3S, 06903 Sophia
  Antipolis Cedex, France}	
\email{Elisabetta.DE-MARIA@univ-cotedazur.fr}  



\begin{abstract}
Formal verification has become increasingly important because of the kinds of guarantees that it can provide for software systems.
Verification of models of biological and medical systems is a promising application of formal verification.
Human neural networks have recently been emulated and studied as a biological system. 
In this paper, we provide a model of some crucial neuronal circuits, called \emph{archetypes}, in the Rocq Prover and prove properties concerning their dynamic behavior.
Understanding the behavior of these modules is crucial because they constitute the elementary building blocks of bigger neuronal circuits.
We consider seven fundamental archetypes (simple series, series with multiple outputs, parallel composition, positive loop, negative loop, inhibition of a behavior, and contralateral inhibition), and prove an important representative property for six of them.
In building up to our model of archetypes, we also provide a general model of \textit{neuronal circuits}, and prove a variety of general properties about neurons and circuits.
In addition, we have defined our model with a longer term goal of modelling the composition of basic archetypes into larger networks, and structured our libraries with definitions and
lemmas useful for proving the properties in this paper as well as those to be proved as future work.
\end{abstract}

\maketitle

\section{Introduction}
\label{sec:intro}
In this work, we apply formal reasoning techniques to the verification of the dynamic behavior of biological human neural networks.
We focus on the level of micro-circuits, and in particular, study \emph{neuronal archetypes}, which are the most elementary circuits, consisting of a few neurons fulfilling a specific computational function.
These archetypes can be coupled to create the elementary building blocks of bigger neuronal circuits.
As an example of a well-known archetype, locomotive motion and other rhythmic behaviors are controlled by specific neuronal circuits called Central  Pattern Generators (CPG) \cite{Matsuoka87BC}, which rely on \emph{contralateral inhibition} (see Section~\ref{subsec:archetypes}).
They can be shown to have various oscillation properties under specific conditions at the circuit level.

The field of systems biology is a recent application area for formal methods, and such techniques have turned out to be very useful so far in this domain \cite{GH15TCS}. 
By modelling and proving properties of a biological system, we open up the potential for deeper understanding of behaviour, disease, effects of medicine, external problems, environmental change impacts, and system recovery of a biological system.
With regard to our particular focus on archetypes, it would be extremely difficult to prove the properties that are expected to hold based on the biological theory through real biological experiments. Exploiting formal methods, and in particular, theorem proving, is an important novel aspect of our work.

As far as the modelling of biological systems is concerned, one approach is to model such systems as graphs whose nodes represent the different possible configurations of a system and whose edges encode meaningful configuration changes.
It is then possible to define and prove properties concerning the temporal evolution of the biological species involved in the system~\cite{FSC04JBPC,RCB04CMSB}.
This often allows deep insight into the biological system at issue, in particular concerning the biological transitions governing it, and the reactions the system will have when confronted with external factors such as disease, medicine, and environmental changes~\cite{DFRS11TCS,TK17CMSB}.
Overall, the literature includes both qualitative and quantitative approaches to model biological systems \cite{wileybook}.
To express the qualitative nature of dynamics, some commonly used formalisms include Thomas' discrete models~\cite{TTK95BMB}, Petri nets~\cite{RML93ISMB}, $\pi$-calculus~\cite{RSS01PSB}, bio-ambients~\cite{RPSCS04TCS}, and reaction rules~\cite{CCDFS04TCS}.
To capture the dynamics from a quantitative point of view, ordinary or stochastic differential equations have been used extensively.
More recent approaches include hybrid Petri nets~\cite{HT98ISB} and hybrid automata~\cite{ABIKMPRS01HSCC}, stochastic $\pi$-calculus~\cite{PC07CMSB}, and rule-based languages with continuous/stochastic dynamics such as Kappa~\cite{DL04TCS}.
Relevant properties concerning the obtained models are then often expressed using a formalism called temporal logic and verified thanks to model checkers such as NuSMV~\cite{CCGR99CAV}, SPIN \cite{H04Addison} or PRISM~\cite{KNP11CAV}.

In~\cite{DDF14FMMB}, the authors propose the use of modal linear logic as a unified framework to encode both biological systems and temporal properties of their dynamic behavior.
They focus on a model of the P53/Mdm2 DNA-damage repair mechanism and they prove some desired properties using theorem proving techniques.
In~\cite{RHST17PO}, the authors advocate the use of higher-order logic to formalize reaction kinetics and exploit the HOL Light theorem prover to verify some reaction-based models of biological networks.
As another example, the Porgy system is introduced in~\cite{FKPINAUD19MSCS}.
It is a visual environment which allows modelling of biochemical systems as rule-based models.
Rewriting strategies are used to choose the rules to be applied.

As far as human neural networks are concerned, there is recent work that has focused on their formal verification.
In~\cite{DLGRG17CSBIO,DMGRG16HSB}, the authors consider the synchronous paradigm to model and verify some specific graphs composed of a few biological neurons.
They introduce and consider most of the archetypes that we cover here.
In that work, some model checkers such as Kind2~\cite{kind2} are employed to automatically verify properties concerning the dynamics of six basic archetypes and their coupling.
However, model checkers prove properties for some given parameter intervals, and do not handle inputs of arbitrary length.
In our work, we use the Rocq Prover~\cite{BC04book,CoqRefManual} to prove important properties of neurons and archetypes.
Rocq implements a highly expressive higher-order logic in which we can directly introduce datatypes modelling neurons and archetypes, and express properties about them.
As a matter of fact, one of the main advantages of using Rocq is the generality of its proofs.
Using such a system, we can prove properties about arbitrary values of parameters, such as any length of time, any input sequence, or any number of neurons.
We use Rocq’s general facilities for structural induction and case analysis, as well as Rocq’s standard libraries that help in reasoning about rational numbers and functions on them.
We believe the approach presented in this paper for reasoning about neural networks is very promising, because it can be exploited for the verification of other kinds of biological networks, such as gene regulatory, metabolic, or environmental networks.

A long term goal of this work is to identify a complete set of circuits at the level of archetypes from which all larger circuits in the brain can be modeled by composing these basic building blocks.
From an electronic perspective, we consider archetypes as biologically inspired logical operators, which are easily adjustable by playing with very few parameters.
Here, we take a significant step toward our goal, and prove a number of properties that have been identified during extensive discussions with neurophysiologists~\cite{DLGRG17CSBIO,DMGRG16HSB}.
This paper can be considered as an extended journal version of~\cite{BDF18CSBIO}, where we presented the first properties and their proofs about single-input neurons, which are simple neurons with only one input.
We also extended one of these properties to the \emph{simple series} archetype.
In~\cite{DBLFGRG22FCS}, we proved an additional property of single-input neurons in the context of comparing our theorem proving approach to model checking.
In~\cite{DDFLOB23ISTE}, we also discussed some of these properties within the more general context of our work on computational logic applied to systems biology.
In terms of proving properties in Rocq about archetypes, those papers considered only one archetype, while here we consider six, including properties that are significantly more complex.
The main contributions of this paper are as follows.
\begin{itemize}
\item We present a formal model of neurons and \emph{archetypes} in Rocq; this model is more general and flexible than those in our previous work~\cite{BDF18CSBIO,DBLFGRG22FCS}.
We define the model in stages.
    \begin{itemize}
    \item We define neurons and important properties such as equivalence and change of state over time.
    \item We model \emph{neuronal circuits}, which include a set of neurons and their internal connections, as well as external connections to sources of input.
    \item We define archetypes as a special form of circuit.
    \end{itemize}
\item We prove important representative properties of the majority of these archetypes.
\item In doing so, we build a large Rocq library of definitions and lemmas about neurons and neuronal circuits, which are crucial for proving these properties and for achieving our future work goals.
\end{itemize}

The paper is organized as follows.
In Section~\ref{sec:background}, we introduce the state of the art relative to neural network modelling, we describe the computational model we have chosen, the Leaky Integrate-and-Fire model (LI\&F), and we briefly introduce the basic archetypes that we consider.
In Section~\ref{sec:model}, we introduce the Rocq Prover and present our model of neuronal circuits in Rocq, which includes definitions of neurons and operations on them, as well as definitions for combining them into general circuits and specific archetypes.
In Section~\ref{sec:neuronprops}, we present and discuss several important properties of single neurons, starting with properties of multiple-input neurons and the relation between the input and output, and then considering the particular case of single-input neurons.
In Section~\ref{sec:circuitprops}, we present some properties of general neuronal circuits.  Because archetypes are circuits, these properties will hold of all archetypes also.
In Section~\ref{sec:archetypeprops}, we present properties of archetypes, moving toward more complex properties that express interactions between neurons and behaviors of neuronal circuits in general.
Finally, in Section~\ref{sec:concl}, we conclude and discuss future work.
The accompanying Rocq code can be found at https://github.com/afelty/NeuronalArchetypesAppendix.git.\footnote{The files currently run in Rocq V9.1.0.}

\section{Background}
\label{sec:background}

We start in Section~\ref{subsec:neural} with some background on neural network modelling.
In Section~\ref{subsec:LIF}, we present the details of the model we use in this work, the Leaky Integrate-and-Fire model (LI\&F).
Finally, in Section~\ref{subsec:archetypes}, we present the seven basic neuronal circuits, or \emph{archetypes}, that we consider.

\subsection{Neural and Neuronal Network Modelling}
\label{subsec:neural}
Neurons are the smallest unit of a neural network~\cite{doi:10.1086/504005}.
They are basically just a single cell.
We can consider them simply as a function with one or more inputs and a single output.
A human neuron receives its inputs via its dendrites.
Dendrites are short extensions connected to the neuron body, which is called a soma.
Inputs are provided in the form of electrical pulses (spikes).
For each neuron there is another extension, called the axon, which plays the role of output.
This extension is also connected to the cell body, but it is longer than the dendrites.
Each neuron has its own activation threshold which is coded somehow inside the soma. The dynamics of each neuron is
characterized through its (membrane) potential value,
which represents the difference of electrical potential
across the cell membrane. The potential value depends on the current input spikes received by the neuron through its dendrites, as well as the decayed value of the previous potential.
When the potential passes its threshold, the neuron fires a spike in the axon. 
Neurons can be connected to other neurons.
Connections happen between the axon of a neuron and a dendrite of another neuron.
Theses connections are called synaptic connections and the location of the connection is called a synapse.
They are responsible for transmitting signals between neurons.

In this paper, we consider third generation models of neural networks.
They are called spiking neural networks~\cite{Maass97,Moisy12,S22IEEE} and they are known for their functional similarity to the biological neural network \cite{YV22BS}.
These biologically realistic models of neurons can carry out efficient computations and are widely employed in the fields of multimedia (tasks such as image classification \cite{BouletIJCNN19}, object detection \cite{Miramond24, SCB23BC}, sound classification \cite{Wu2018}), and robotics \cite{Marrero2024}.
They have been proposed in the literature with different complexities and capabilities.
In this work we focus on the Leaky Integrate-and-Fire (LI\&F) model originally proposed in~\cite{lapicque1907}.
It is a computationally efficient approximation of a single-compartment model~\cite{izhikevich04} and is abstract enough to be able to apply formal verification techniques.
In such a model, neurons integrate present and past inputs in order to update their membrane potential values.
Whenever the potential exceeds a given threshold, an output signal is fired.

As far as spiking neural networks are concerned, in the literature there are a few attempts at giving formal models for them.
In~\cite{AC16TCS}, a mapping of spiking neural P systems into timed automata is proposed.
In that work, the dynamics of neurons are expressed in terms of evolution rules and durations are given in terms of the number of rules applied.
Timed automata are also exploited in~\cite{DD18BIOSTEC, naco20} to model LI\&F networks.
This modelling is substantially different from the one proposed in~\cite{AC16TCS} because an explicit notion of duration of activities is given.
Such a model is formally validated against some crucial properties defined as temporal logic formulas and is then exploited to find an assignment for the synaptic weights of neural networks so that they can reproduce a given behavior.
In \cite{Partha2023}, spiking neural networks are soundly mapped into timed automata and several biologically plausible behaviors of individual spiking neurons are formally verified for three classes of spiking models (LI\&F, Quadratic Integrate and Fire, and Izhikevich).
 
The work mentioned earlier that employs model checking~\cite{DLGRG17CSBIO,DMGRG16HSB} also uses the LI\&F model, in this case modelling LI\&F neurons and some basic small circuits using the synchronous language Lustre.
Such a language is dedicated to the modelling of reactive systems, i.e., systems which constantly interact with the environment and which may have an infinite duration.
It relies on the notion of logical time: time is considered as a sequence of discrete instants, and an instant is a point in time where external input events can be observed, computations can be done, and outputs can be emitted.
Lustre is used not only to encode neurons and some basic archetypes (simple series, parallel composition, etc.), but also some properties concerning their dynamic evolution.
As mentioned, Kind2 was employed, and in particular, these properties were automatically proved for some given parameter intervals.

LI\&F networks extended with probabilities are formalized as discrete-time Markov chains in~\cite{DGRR18BIOSTEC}.
The proposed framework is then exploited to propose an algorithm which reduces the number of neurons and synaptic connections of input networks while preserving their dynamics.

\subsection{Discrete Leaky Integrate-and-Fire Model}
\label{subsec:LIF}
\ \

In this section, we introduce a discrete (Boolean) version of LI\&F modeling.
Notice that discrete modeling is well suited to this kind of model because neuronal activity, as with any recorded physical event, is only known through discrete recording (the recording sampling rate is usually set at a significantly higher resolution than the one of the recorded system, so that there is no loss of information).
We first present the basic biological knowledge associated to the modeled phenomena and then we detail the adopted model.

When a neuron receives a signal at one of its synaptic connections, it produces an excitatory or an inhibitory \textit{post-synaptic potential} (PSP) caused by the opening of selective ion channels according to the post-synaptic receptor nature.
An inflow of cations in the cell leads to an activation; an inflow of anions in the cell corresponds to an inhibition.
This local ions flow modify the membrane potential either through a depolarization (excitation) or a hyperpolarization (inhibition).
Such variations of the membrane potential are progressively transmitted to the rest of the cell.
The potential difference is called \textit{membrane potential}.
In general, several to many excitations are necessary for the membrane potential of the post-synaptic neuron to exceed its \textit{depolarization threshold}, and thus to emit an \textit{action potential} at its axon hillock to transmit the signal to other neurons.

Two phenomena allow the cell to exceed its depolarization threshold: the \textit{spatial summation} and the \textit{temporal summation} \cite{doi:10.1086/504005}.
Spatial summation allows to sum the PSPs produced at different areas of the membrane.
Temporal summation allows to sum the PSPs produced during a finite time window.
This summation can be done thanks to a property of the membrane that behaves like a capacitor and can locally store some electrical loads (\textit{capacitive property}).

The neuron membrane, due to the presence of leakage channels, is not a perfect conductor and capacitor and can be compared to a resistor inside an electrical circuit.
Thus, the range of the PSPs decreases with time and space (\emph{resistivity} of the membrane).

A LI\&F neuronal network is represented with a weighted directed graph where each node stands for a neuron soma and each edge stands for a synaptic connection between two neurons.
The associated weight for each edge is an indicator of the weight of the connection on the receiving neuron: a positive (resp. negative) weight is an activation (resp. inhibition).

The depolarization threshold of each neuron is modeled via the \textit{firing threshold } $\tau$, which is a numerical value that the neuron membrane potential $p$ shall exceed at a given time $t$ to emit an action potential, or \textit{spike}, at the time $t+1$.

The membrane resistivity is symbolized with a numerical coefficient called the \textit{leak factor} $r$, which allows to decrease the range of a PSP over time.

Spatial summation is implicitly taken into account.
In our model, a neuron $u$ is connected to another neuron $v$ via a single synaptic connection of weight $w_{uv}$.
This connection represents the entirety of the shared connections between $u$ and $v$.
Spatial summation is also more explicitly taken into account with the fact that, at each instant, the neuron sums each signal received from each input neuron.
The temporal summation is done through a sliding integration window of length $\sigma$ for each neuron to sum all PSPs.
Older PSPs are decreased by the leak factor $r$.
This way, the biological properties of the neuron are respected and the computational load remains limited.
This allows us to obtain finite state sets, and thus to easily apply model checking techniques.

More formally, the following definition can be given:

\begin{defi}\textbf{Boolean Spiking Integrate-and-Fire Neural Network.}
			\label{def.snn}
			\emph{A spiking Boolean integrate and fire neural network} is a tuple $(V,\, E,\, w)$, where: \begin{itemize}
				\item $V$ are Boolean spiking integrate and fire neurons,
				\item $E \subseteq V \times V$ are synapses,
				\item $w: E \rightarrow \mathbb{Q} \cap [-1,1]$ is the synapse weight function associating to each synapse $(u,\, v)$ a weight $w_{uv}$.
			\end{itemize}
            \emph{A spiking Boolean integrate and fire neuron} is a tuple $(\tau, r, p, y)$, where:
            \begin{itemize}
				\item $\tau \in \mathbb{Q}^+$ is the \emph{firing threshold},
                \item $r \in \mathbb{Q} \cap [0, 1]$ is the \emph{leak factor},

                \item $p: \mathbb{N} \rightarrow \mathbb{Q}$ is the [membrane] \emph{potential} function defined as

                \[
p(t) =\left\{ \begin{array}{ll}
\sum_{i=1}^{m} w_i \cdot x_i(t) & \mathrm{if\ } p(t-1) \geqslant \tau\\
\sum_{i=1}^{m} w_i \cdot x_i(t) + r \cdot p(t - 1) \quad &\mathrm{otherwise}
\end{array}\right.
\]
		where $p(0)=0$, $m$ is the number of inputs of the neuron, $w_i$ is the weight of the synapse connecting the $i^{th}$ input neuron to the current neuron, and  $x_i(t) \in \{0, 1\}$ is the signal received at the time $t$ by the neuron through its $i^{th}$ input synapse (observe that, after the potential exceeds its threshold, it is reset to $0$),
                \item $y: \mathbb{N} \rightarrow \{0, 1\}$ is the neuron output function,  defined as
    \[
			y(t) = \left\{ \begin{array}{ll}
				1 \quad & \mathrm{if\ }p(t) \geqslant \tau \\
				0 & \mathrm{otherwise}.
			\end{array}\right.
                        \]
	    \end{itemize}	
\end{defi}

\subsection{Archetypes}
\label{subsec:archetypes}

As mentioned, in neural networks, it is possible to identify some mini-circuits with a relevant topological structure.
Each one of these small modules, which are often referred to as archetypes in the literature (e.g.,~\cite{DMGRG16HSB}), displays a specific class of behaviors.
They are shown in Figure~\ref{fig:archetypes}.
\begin{figure}[htb]
\includegraphics[scale=.45]{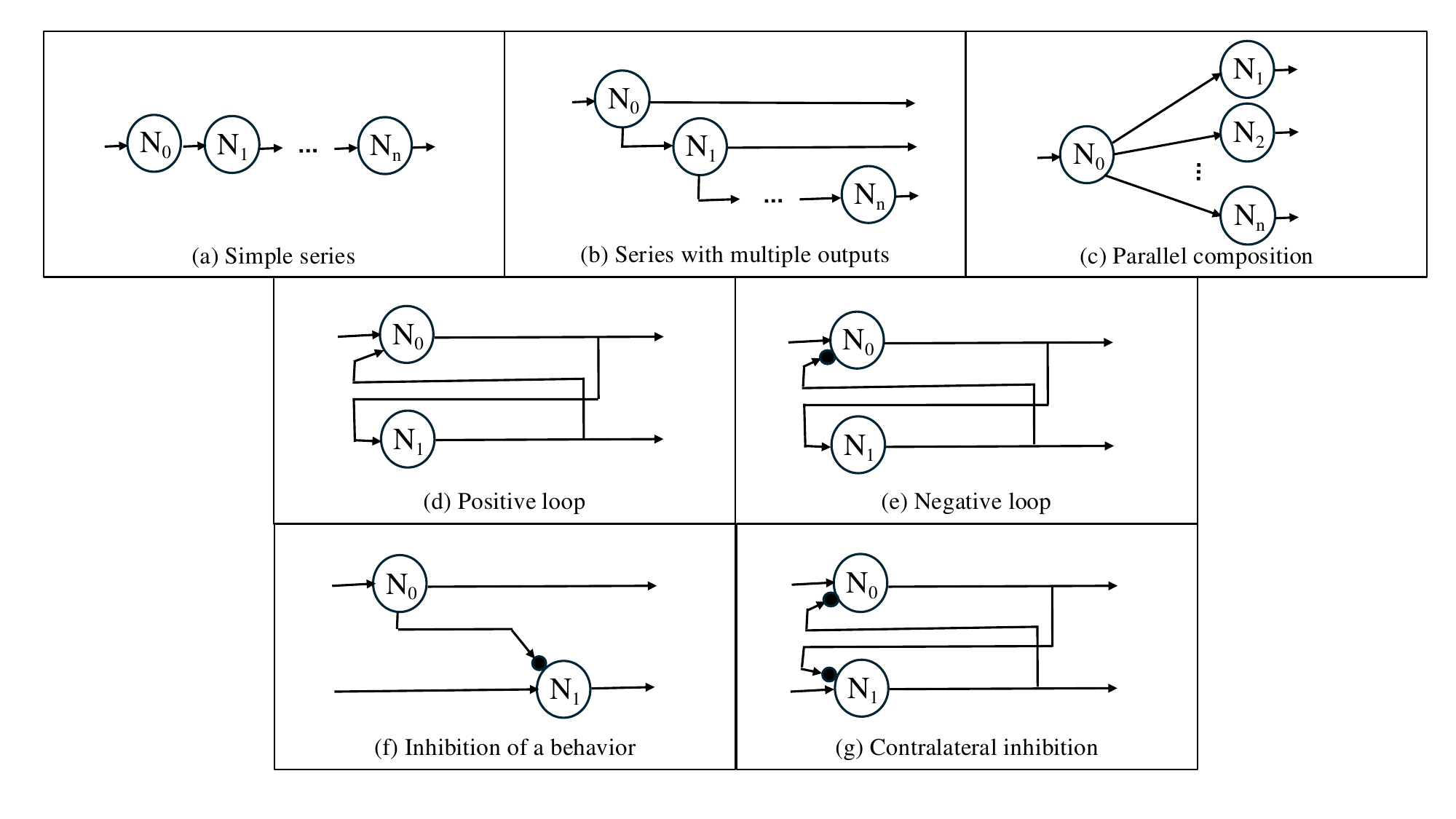}
\caption{The basic neuronal archetypes}
\label{fig:archetypes}
\end{figure}
We have added an additional one, the \emph{positive loop}, appearing as (d) in the figure, since it is also relevant and has its own interesting properties.
\begin{enumerate}[label=(\alph*)]
\item {\emph{Simple series}} is a sequence of neurons where each element of the chain receives as input the output of the preceding one.
\item {\emph{Series with multiple outputs}} is a series where, at each time unit, we are interested in knowing the outputs of all the neurons (i.e., all the neurons are considered as output neurons).
\item {\emph{Parallel composition}} is a set of neurons receiving as input the output of a given neuron.
\item {\emph{Positive loop}} is a loop consisting of two neurons: the two neurons activate each other.
\item {\emph{Negative loop}} is a loop consisting of two neurons: the first neuron activates the second one while the latter inhibits the former.
\item {\emph{Inhibition of a behavior}} consists of two neurons, the first one inhibiting the second one.
\item {\emph{Contralateral inhibition}} consists of two or more neurons, each one inhibiting the other(s).
\end{enumerate}
The arrows entering neuron $N_0$ in the archetypes can be considered as output coming from a neuron outside the circuit, i.e., some neuron other than the ones shown in the figure.  We call these neurons \emph{external sources of input}.  Note that the inhibition and contralateral inhibition archetypes have a second external source of input linked to $N_1$.  
We define the \emph{environment of a ciruit} to be all of the neurons in the circuit plus all of the neurons serving as external sources of input.

\section{Modelling Neurons and their Properties in Rocq}
\label{sec:model}

The Rocq Prover is an interactive theorem prover that implements a highly expressive logic called the Calculus of Inductive Constructions (CIC)~\cite{BC04book}.
It has been widely used in the verification of programs and systems.
Proofs of correctness of computer systems and (in our case) biological systems involve many technical details that are easy to get wrong; formal proof helps to automate the repetitive cases as well as guarantee that no cases are omitted.

Our development does not depend on advanced features of Rocq like dependent types or the hierarchy of universes, and thus while we take into account everything about the model and properties, our model is not difficult to understand, even for readers with a basic knowledge of theorem proving, and could likely be translated fairly easily to other interactive theorem provers such as Nuprl~\cite{Nuprl86}, the PVS Specification and Verification System \cite{pvs}, and Isabelle/HOL \cite{isabelle}.

\subsection{Introduction to Rocq}
\label{coqintro}
In this section, we introduce the main Rocq features that we exploit for neural network modeling.
More complete documentation of Rocq can be found in~\cite{BC04book,CoqRefManual}.
Expressions in CIC include a functional programming language.
This language is typed (every Rocq expression has a type).
For instance, \texttt{X:nat} says that variable \texttt{X} takes its value in the domain of natural numbers.
The types we employ in our model include \texttt{nat}, \texttt{Q}, \texttt{bool}, and \texttt{list} which denote natural numbers, rational numbers, booleans, and lists of elements respectively.
These types are available in Rocq's standard libraries.
All the elements of a list must have the same type.
For example, \texttt{L:list nat} expresses that \texttt{L} is a list of natural numbers.
An empty list is denoted by \texttt{[]} or \texttt{nil} in Rocq.
Functions are the basic components of functional programming languages. The general form of a Rocq function is shown in Figure~\ref{fig:Coqfunction}.
\begin{figure}[htb]
\begin{tabular}{l}
\texttt{Definition}/\texttt{Fixpoint} Function\_Name\\
\quad (Input$_1$: Type of Input$_1$) $\dots$
      (Input$_n$: Type of Input$_n$): Output Type :=\\
\quad\quad Body of the function.
\end{tabular}
\caption{General form of a function definition in Rocq}
\label{fig:Coqfunction}
\end{figure}
The Rocq keywords \texttt{Definition} and \texttt{Fixpoint} are used to define non-recursive and recursive functions, respectively.
These keywords are followed by the function name.
After the function name, there are the input arguments and their corresponding types.
Inputs having the same type can be grouped.
For instance, \texttt{(X Y Z: Q)} states all variables \texttt{X}, \texttt{Y}, and \texttt{Z} are rational numbers.
Inputs are followed by a colon, which is followed by the output type of the function.
Finally, there is a Rocq expression representing the body of the function, followed by a dot.

In Rocq, pattern matching is exploited to perform case analysis.
This useful feature is used, for instance, in recursive functions, for discriminating between base cases and recursive cases.
For example, it is employed to distinguish between empty and nonempty lists.
A non-empty list consists of the first element of the list (the \emph{head}), followed by a double colon, followed by the rest of the list (the \emph{tail}).
The tail of a list itself is a list and its elements have the same type as the head element.
For instance, let $L$ be the list \texttt{(5::2::9::nil)} containing three natural numbers.
In Rocq, the list $L$ can also be written as \texttt{[5;2;9]}, where the head is \texttt{5} and the tail is \texttt{[2;9]}.
Thus, non-empty lists in Rocq often follow the general pattern \texttt{(h::t)}.
In addition, there are two functions in the Rocq library called \texttt{hd} and \texttt{tl} that return the head and the tail of a list, respectively.
For example, \texttt{(hd d l)} returns the head of the list \texttt{l}.
Here, \texttt{d} is a default value returned if \texttt{l} is an empty list and thus does not have a head.
Also, \texttt{(tl l)} returns the tail of the list \texttt{l} and returns \texttt{nil} if there is no tail.

A natural number (\texttt{nat}) is either \texttt{0} or the successor of another natural number, written \texttt{(S n)}, where \texttt{n} is a natural number.
For instance, \texttt{1} is represented as \texttt{(S 0)}, \texttt{2} as \texttt{(S (S 0))}, etc.
In Figure~\ref{fig:patterns}, some patterns for lists and natural numbers are shown using Rocq's \texttt{match$\dots$with$\dots$end} pattern matching construct.
\begin{figure}[htb]
  \begin{tabular}{l}
    \texttt{match X with}\\
    $\mid$ \texttt{0} $\Rightarrow$ calculate something when
      \texttt{X} = \texttt{0}\\
    $\mid$ \texttt{S n} $\Rightarrow$ calculate something when
      \texttt{X} is successor of \texttt{n}\\
      \texttt{end}\\ \\
      \texttt{match L with}\\
      $\mid$ \texttt{[]} $\Rightarrow$ calculate something when
      \texttt{L} is an empty list\\
      $\mid$ \texttt{h::t} $\Rightarrow$ calculate something when
      \texttt{L} has head \texttt{h} followed by tail \texttt{t}\\
      \texttt{end}\\
  \end{tabular}
  \caption{Rocq pattern matching forms for natural numbers and lists}
  \label{fig:patterns}
\end{figure}

Other examples of definitions and functions found in Rocq libraries that we use in this paper include the \texttt{if} statement on booleans, the \texttt{length}, \texttt{map}, and $\mathop{++}$ (append), and \texttt{In} (list membership) functions on lists.

In addition to the data types that are defined in Rocq libraries, new data types can be introduced.
One way to do so is using Rocq's facility for defining records.
Records can have different fields with different types.
For instance, we can define a record with three fields \texttt{Fieldnat}, \texttt{FieldQ}, and \texttt{ListField}, which have types natural number, rational number, and list of natural numbers, respectively, as illustrated in Figure~\ref{fig:record}.
\begin{figure}[htb]
\begin{verbatim}
          Record Sample_Record := MakeSample {
            Fieldnat: nat;
            FieldQ: Q;
            ListField: list nat;
            CR: Fieldnat > 7 }.

         S: Sample_Record
\end{verbatim}
\caption{Example definition and variable declaration for a record type}
\label{fig:record}
\end{figure}
Fields in Rocq records can also represent constraints on other fields.
For instance, field \texttt{CR} in Figure~\ref{fig:record} expresses that field \texttt{Fieldnat} must be greater than 7.

A record is a type like any other type in Rocq, and so variables can have the new record type.
For example, variable \texttt{S} with type \texttt{Sample\_Record} in the figure is an example.
When a variable of this record type gets a value, all the constraints in the record have to be satisfied.
For example, \texttt{Fieldnat} of \texttt{S} cannot be less than or equal to 7.

\subsection{Defining Neurons and their Properties in Rocq}

We start illustrating our formalization of neuronal networks in Rocq with the code in Figure~\ref{fig:Coqneuron}.
\begin{figure}[htb]
\begin{verbatim}
Record NeuronFeature :=
  MakeFeat {
      Id : nat;
      Weights : nat -> Q; 
      Leak_Factor : Q;
      Tau : Q;
      LeakRange : 0 <= Leak_Factor <= 1;
      PosTau : 0 < Tau;
      WRange : forall x, -1 <= Weights x <= 1;
      WId : Weights Id = 0;
    }.

Record Neuron :=
  MakeNeuron {
      Output : list bool;
      CurPot : Q; 
      Feature : NeuronFeature;
      CurPot_Output : (Tau Feature <=? CurPot) = hd false Output;
    }.

Definition SetNeuron (nf : NeuronFeature): Neuron:=
  MakeNeuron [false] 0 nf (setneuron_curpot_output _).
\end{verbatim}
\caption{Rocq definition of a neuron}
\label{fig:Coqneuron}
\end{figure}
To define a neuron, we use Rocq's basic record structures. At this level, the neuron is considered within an environment of other neurons but not yet as part of a circuit. The concept of boundaries between neurons within the circuit and external input sources is not introduced at this stage; this notion is discussed in Section~\ref{subsec:circuits-props}.  We define it in two parts.  
The first part, \texttt{NeuronFeature}, consists of eight fields with their corresponding types.
The first four fields contain the data that models a neuron.
The \texttt{Id} field is a natural number identifier. Each neuron in an environment will have a unique identifier.
The weights linked to the inputs of the neuron (\texttt{Weights}) are defined as a function from identifiers to rational numbers, where each input is the identifier of a neuron in the environment.
As we will see, when a circuit has $n$ neurons, the identifiers will be in the range $0,\ldots,n-1$, the identifiers of the external sources will be consecutive number(s) starting at $n$, and all other arguments to the \texttt{Weights} function will not be used.
The next two fields,  \texttt{Leak\_Factor} and \texttt{Tau} represent the leak factor and firing threshold, respectively; they are both rational numbers.

The next four fields in the \texttt{NeuronFeature} record represent constraints that the first four fields must satisfy according to the LI\&F model defined in Section~\ref{subsec:LIF}.
The first three conditions---\texttt{LeakRange}, \texttt{PosTau}, and \texttt{WRange}---are the Rocq representation for the conditions $r \in \mathbb{Q} \cap [0, 1]$, $\tau \in \mathbb{Q}^+$, $w: E \rightarrow \mathbb{Q} \cap [-1,1]$, respectively, from Definition \ref{def.snn}.
Since we do not consider neurons that have self-connections,\footnote{In some cases a neuron can be connected to itself, self-connections are called autapses. However their exact function is not fully understood and biologists do not often consider these connections.} the last condition (\texttt{WId}) assigns a default weight of $0$ to the neuron's own identifier.

The \texttt{NeuronFeature} record just described represents the static information about a neuron.
The next definition in Figure~\ref{fig:Coqneuron} is the \texttt{Neuron} record, which has four fields, one of which is a \texttt{NeuronFeature},
two that record dynamic information representing the state of a neuron, and one that imposes a constraint on the dynamic information.
The \texttt{Output} records the output of the neuron since time $0$ up until the current time.  Every output is $0$ (no spike) or $1$ (spike), and we use booleans to represent them.
There is one entry for each time step, represented in reverse order;
the last element in the list is the output at time $0$.
Storing the values in reverse order helps make proofs by induction over lists easier in Rocq.
\texttt{CurPot} stores the most recent membrane potential.
\texttt{Output} and \texttt{CurPot} store the values computed by the $y$ and $p$ functions, respectively, from Definition~\ref{def.snn}. Rocq functions that perform these calculations will be defined below.
\texttt{CurPot\_Output} is a constraint that indicates whether or not a neuron fires at the current time, as defined in Definition~\ref{def.snn}. The head element of \texttt{Output} represents the output at the current time and is a boolean value indicating whether the \texttt{CurPot} value has reached the threshold, defined by \texttt{(Tau Feature)}.

We make a few remarks about our choices of representation.
Note that rational numbers $\mathbb{Q}$ are central to Definition~\ref{def.snn}.
In Rocq, there are several choices or representing them.  
We chose the type \texttt{Q} because the accompanying Rocq library is well-developed; we added only a small number of lemmas that helped to simplify our proofs.  
In addition, we were able to use this library to our advantage in developing formal proofs that closely followed the informal ones.  
Note also that we represent some functions from Definition~\ref{def.snn}, such as the weights function $w$, as Rocq functions, while others, such as the output function $y$, are represented using lists.  We follow the standard practice and use lists whenever they are convenient for representing finite parts of functions.
With regard to the weights function, another possible representation would be as a matrix, especially in the case of multiple-input neurons.  
We decided against such a representation because matrix multipication would complicate some of the formal reasoning, especially with regard to our longer term goal of modelling the composition of basic archetypes into larger networks.

If \texttt{NF} is a \texttt{NeuronFeature}, \texttt{(Id NF)}  denotes its first field, and similarly for the other fields and records.
A new neuron with values \texttt{id}, \texttt{W}, \texttt{L}, and \texttt{T} for the fixed fields and \texttt{O}, \texttt{C}, and \texttt{C\_O} for the dynamic fields, all with the suitable types, and proofs $\mathtt{P1},\dots,\mathtt{P4}$ of the four constraints, is written \texttt{(MakeNeuron O C (MakeFeat id W L T P1 P2 P3 P4) C\_O)}.
This notation will appear in Rocq code.
In presenting examples, properties, and proofs, we will adopt several conventions for clarity.  If \texttt{N} is a neuron, we write $\mathit{id}_N$ to denote \texttt{(Id (Feature N))}, and similarly $w_N$, $\mathit{lk}_N$, $\tau_N$ for the other three fields, respectively, that represent the static information of a neuron.  Recall that $w_N$ is a function.
We use $w_N(id)$ to denote \texttt{(Weights (Feature N) id)} where \texttt{id} is the identifier of some neuron in the environment.
Finally, given neuron \texttt{N},
we denote \texttt{(Output N)}, \texttt{(Feature N)}, \texttt{(CurPot N)}, and \texttt{(CurPot\_Output N)} as $\mathit{Output}(N)$, $\mathit{Feature}(N)$, $\mathit{CurPot}(N)$, and $\mathit{CurPot\_Output(N)}$, respectively.

\begin{exas}
Consider an example simple series archetype from Figure~\ref{fig:archetypes}(a) of length 3 with neurons denoted $N_0$, $N_1$, and $N_2$ with the external input coming from some neuron $N_3$.  
The subscripts represent the value of the \texttt{Id} field of these neurons.
The weight $w_{N_1}(0)$ is the weight assigned to the edge between $N_0$ and $N_1$.  Similarly $w_{N_2}(1)$ is the weight on the edge between $N_1$ and $N_2$, and $w_{N_0}(3)$ is the weight on the external input.
All other values of the weight functions with arguments in the range $0,\ldots,3$ will be $0$.
Leaving out the values outside this range, the weight functions are:
$$\begin{array}{rcllll}
w_{N_0} & = & \{0\mapsto 0,&1\mapsto 0,&2\mapsto 0,&3\mapsto w_{N_0}(3)\}\\
w_{N_1} & = & \{0\mapsto w_{N_1}(0),&1\mapsto 0,&2\mapsto 0,&3\mapsto 0\}\\
w_{N_2} & = & \{0\mapsto 0,&1\mapsto w_{N_2}(1),&2\mapsto 0,&3\mapsto 0\}.
\end{array}$$
Note that for $n=0,\ldots,2$, $w_{N_n}(n)=0$; a neuron is never connected to itself.

Modifying this example slightly, suppose there is also an edge from $N_2$ back to $N_1$ with weight $w_{N_1}(2)$.  In this case, $w_{N_1}$ would become:
$$\begin{array}{rcllll}
w_{N_1} & = & \{0\mapsto w_{N_1}(0),&1\mapsto 0,&2\mapsto w_{N_1}(2),&3\mapsto 0\}.
\end{array}$$
Figure~\ref{fig:example1} contains a diagram illustrating this example.  Weights associated with each neuron are displayed on the incoming arrows. The modification is illustrated with a dotted arrow.
\label{exas:weights}
\end{exas}

\begin{figure}[htb]
\vspace{-2.3cm}
\includegraphics[scale=.5,angle=90]{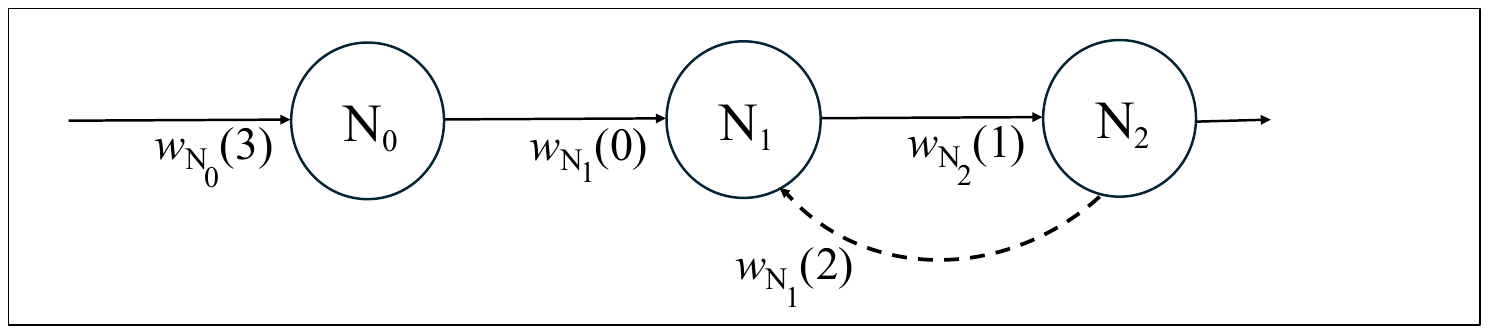}
\vspace{-6.3cm}
\caption{Example archetype with weights}
\label{fig:example1}
\end{figure}

Figure~\ref{fig:Coqneuron} contains one more definition.
The \texttt{SetNeuron} function takes any \texttt{NeuronFeature} and creates an \textit{initial} neuron at time $0$, with only one element in the output, which represents $0$ at time $0$, and a default current potential of $0$.
In our model, we also have a notion of a well-formed neuron.\footnote{See definition \texttt{well\_formed\_neuron} in the Rocq code.}
A neuron is \textit{well-formed} when its output at time $0$ is $0$, i.e., the last element in the output list is \texttt{false}.
When we define functions on neurons, we include lemmas stating that well-formedness is preserved.
We could have included well-formedness information directly as a constraint inside the \texttt{Neuron} record definition in Figure~\ref{fig:Coqneuron}.
We leave it out in order to keep our model as general as possible.
For example, it may be useful to represent neurons starting at a time other than $0$, or to represent other information such as flow of input and output.  
In such cases, the constraints would be different.

We next consider the Rocq representation of the functions computing the current potential and the output of a neuron from Figure~\ref{def.snn}.
\begin{figure}[htb]
\begin{verbatim}
Fixpoint potential (ws : nat -> Q) (inp : nat -> bool) (len : nat) : Q :=
  match len with
  | O => 0
  | S m => if inp m
           then (potential ws inp m) + ws m
           else (potential ws inp m)
  end.

Definition NextPotential (N : Neuron) (inp : nat -> bool) (len : nat) : Q :=
  if Tau (Feature N) <=? CurPot N
  then (potential (Weights (Feature N)) inp len)
  else (potential (Weights (Feature N)) inp len) + 
       (Leak_Factor (Feature N)) * (CurPot N).

Definition NextOutput (N : Neuron) (p : Q) := Tau (Feature N) <=? p.
\end{verbatim}
\caption{Rocq functions for computing membrane potential and output}
\label{fig:Coqpotential}
\end{figure}
The first definition in Figure~\ref{fig:Coqpotential} computes the weighted sum of the inputs of a neuron, which is fundamental for the calculation of the potential.
In this recursive function, there are three arguments: \texttt{ws} is a \texttt{Weights} function, \texttt{inp} is an input function, which is a function from input identifiers to input values, and \texttt{len} represents the number of  neurons in the environment.
In particular, the \texttt{len} input values are $(\mathtt{inp}~0),\ldots,(\mathtt{inp}~(\texttt{len}-1))$ and the corresponding weights on those inputs are $(\mathtt{ws}~0),\ldots,(\mathtt{ws}~(\texttt{len}-1))$. 
The function returns an element of type \texttt{Q}.
The \texttt{NextPotential} function completes the definition of the $p$ function from Definition \ref{def.snn} for a specific neuron, calling \texttt{potential} and looking up the values of \texttt{Tau}, \texttt{CurPot}, and \texttt{Leak\_Factor} as needed in the calculation.
The \texttt{NextOutput} function is the Rocq definition of the $y$ function from Definition \ref{def.snn}, assuming that the input \texttt{p} is the value of the current potential.

\begin{exas}
We continue Example~\ref{exas:weights} and illustrate the potential function in Figure~\ref{fig:Coqpotential}.
We focus on $N_1$ from Example~\ref{exas:weights}, and
take the second version of $w_{N_1}$, which is the version with the arrow from $N_2$ back to $N_1$, resulting in 2 inputs to $N_1$.  We add some sample input functions also, as follows:
$$\begin{array}{rcllll}
w_{N_1} & = & \{0\mapsto w_{N_1}(0),&1\mapsto 0,&2\mapsto w_{N_1}(2),&3\mapsto 0\}\\
\mathit{inp}_1 & = & \{0\mapsto \mathit{true},&1\mapsto \mathit{true},&2\mapsto \mathit{true},&3\mapsto \mathit{true}\}\\
\mathit{inp}_2 & = & \{0\mapsto \mathit{true},&1\mapsto \mathit{false},&2\mapsto \mathit{true},&3\mapsto \mathit{false}\}\\.
\end{array}$$
We illustrate in Figure~\ref{fig:example2}, adding explicitly to the diagram the leak factor, the threshold and the inputs to $N_1$ from $N_0$ and $N_2$; $\mathit{inp}_1$ and $\mathit{inp}_2$ have the same values coming from these two neurons.
There are 4 neurons in the environment, so in a call to the $\mathit{potential}$ function the argument $\mathit{len}$ will be $4$.  Consider the two calls below, with recursive calls unwound:
$$\begin{array}{rclllllll}
\mathit{potential}(w_{N_1},\mathit{inp}_1,4) & = & w_{N_1}(0) & + & w_{N_1}(1) & + & w_{N_1}(2)& + & w_{N_1}(3)\\
\mathit{potential}(w_{N_1},\mathit{inp}_2,4) & = & w_{N_1}(0) & &  & + & w_{N_1}(2)& & 
\end{array}$$
Since there is no link from $N_1$ to itself, and no link from $N_3$ to $N_1$ in this example, the values of both $w_{N_1}(1)$ and $w_{N_1}(3)$ are $0$, so the value of both calls is $w_{N_1}(0) + w_{N_1}(2)$; in the second call, the values of $w_{N_1}(1)$ and $w_{N_1}(3)$ are not accessed because the values of $\mathit{inp}_2(1)$ and $\mathit{inp_2}(3)$ are both $\mathit{false}$.
\label{exas:potential}
\end{exas}

\begin{figure}[htb]
\vspace{-1.7cm}
\includegraphics[scale=.5,angle=90]{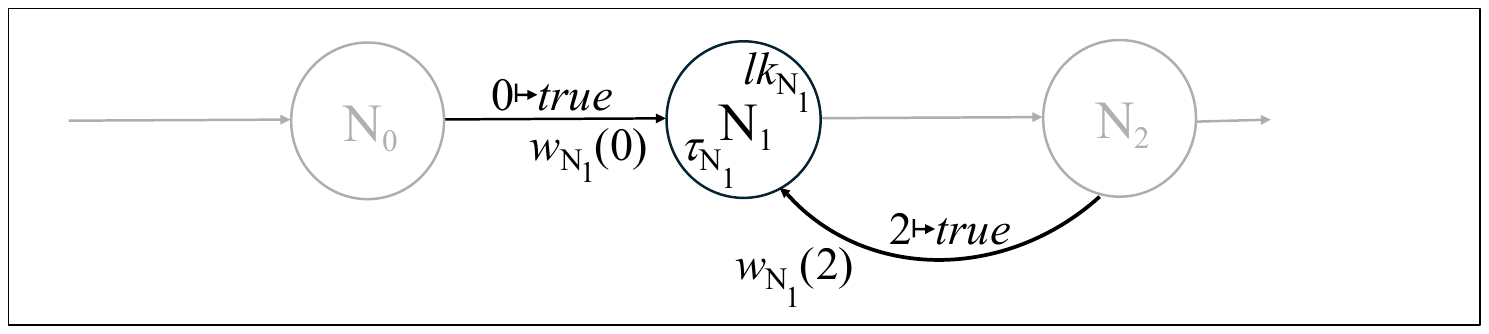}
\vspace{-6.3cm}
\caption{Example data for computing potential for $N_1$}
\label{fig:example2}
\end{figure}

Below is an important property about the \texttt{potential} function in the case when the weights of all identifiers (of neurons) in the environment are non-negative.
In this case, the potential will also be non-negative.
The interested reader can find the corresponding property and proof in the Rocq code, for this and all other properties, using the name given in brackets.
The proofs in the paper follow closely the structure of the Rocq proofs, presenting the main inductions as well as other proof strategies and lemmas.
We adopt several more conventions for readability.
In particular, we use more standard mathematical and logical notation in stating properties and presenting proofs.
For instance, arguments to predicates and functions will be in parentheses, separated by commas, e.g., \texttt{(potential w inp m)} will be written $\mathit{potential}(w,\mathit{inp},\mathit{len})$.
Also, although Rocq uses different syntax for inequalities between elements of different types, we will overload and use the standard mathematical notation, e.g., $\mathit{potential}(w,\mathit{inp},\mathit{len})\geq0$.

\begin{lem}[Always Non-Negative Potential]
\texttt{[potential\_nneg\_w]}\\\ \\
$\begin{array}{l}
  \forall (\mathit{w}:\mathit{nat}\rightarrow\mathit{Q})  (\mathit{inp}:\mathit{nat}\rightarrow\mathit{bool}) (\mathit{len}:\mathit{nat}),\\
  \quad (\forall (\mathit{id}:\mathit{nat}), \mathit{id} < \mathit{len} \rightarrow w(\mathit{id})\geq 0)\rightarrow\\
  \quad\quad \mathit{potential}(w, \mathit{inp}, \mathit{len})\geq 0.
  \end{array}$
\label{lem:nonnegpot}
\end{lem}

\begin{proof} 
The proof proceeds by induction on the structure of the number of sources in the environment $\mathit{len}$.\\
\textbf{Case $\mathit{len} = 0$.} $\mathit{potential}(w, \mathit{inp}, 0)$ is null by definition of the function $\mathit{potential}$. The first case is verified.\\
\textbf{Case $\mathit{len} = S(n)$ (successor of n)}. We assume that the property holds for $n$, and now we aim to show that it also holds for $\mathit{len}$. By definition, \[
\mathit{potential}(w, \mathit{inp}, \mathit{len}) = 
\begin{cases} 
\mathit{potential}(w, \mathit{inp}, n) + w(n), & \mathit{if\ } \mathit{inp}(n) = \mathit{true}, \\
\mathit{potential}(w, \mathit{inp}, n), & \mathit{otherwise}.
\end{cases}
\]

\textbf{Sub-case $\mathit{inp}(n)$ is true}. By the induction hypothesis, $\mathit{potential}(w, \mathit{inp}, \mathit{len})$ is non-negative, since the weights from $0$ to $n-1$ are non-negative.
$w(n)$ is non-negative by the definition. We conclude that $\mathit{potential}(w, \mathit{inp}, \mathit{len})$ is non-negative.

\textbf{Sub-case $\mathit{inp}(n)$ is false}. We conclude that $\mathit{potential}(w, \mathit{inp}, \mathit{len})$ is non-negative by a similar argument as in the previous case.
\end{proof}
Conversely, we have a theorem for cases where the weights are non-positive.
\begin{lem}[Always Non-Positive Potential]
\texttt{[potential\_npos\_w]}\\\ \\
$\begin{array}{l}
  \forall (\mathit{w}:\mathit{nat} ~ \rightarrow\mathit{Q})  (\mathit{inp}: \mathit{nat}\rightarrow\mathit{bool}) (\mathit{len}:\mathit{nat}),\\
  \quad (\forall (\mathit{id}:\mathit{nat}), \mathit{id} < \mathit{len} \rightarrow w(\mathit{id})\leq 0)\rightarrow\\
  \quad\quad \mathit{potential}(w, \mathit{inp}, \mathit{len})\leq0.
  \end{array}$
\label{lem:nonpospot}
\end{lem}
The proof is similar to the one for Lemma~\ref{lem:nonnegpot}.

A neuron changes its state by processing its inputs.
\begin{figure}[htb]
\begin{verbatim}
Definition NextNeuron (inp : nat -> bool) (len : nat) (n : Neuron) :=
  let next_potential := NextPotential n inp len in
  MakeNeuron (NextOutput n next_potential :: Output n)
             next_potential (Feature n) (eq_refl _).

Fixpoint AfterNstepsNeuron (N : Neuron) (inp : list (nat -> bool))
         (len: nat) :=
  match inp with
  | nil => N
  | h::tl => NextNeuron h len (AfterNstepsNeuron N tl len) 
  end.
\end{verbatim}
\caption{Updating the dynamic fields of a neuron}
\label{fig:updateneuron}  
\end{figure}
A single-step state change occurs by applying the \texttt{NextNeuron} function in Figure~\ref{fig:updateneuron} to a neuron, with one value for each of its inputs.
The \texttt{NextPotential} function from Figure~\ref{fig:Coqpotential} is applied, resulting in value \texttt{next\_potential}.  A new neuron is created with \texttt{next\_potential} as the new value of \texttt{CurPot}.  In particular, this value appears as the second argument to \texttt{MakeNeuron}.
More specifically, recall that \texttt{(CurPot N)} used in the body of \texttt{NextPotential} is the most recent value of the current potential of the neuron, or $p(t-1)$; calling \texttt{NextPotential} in the body of \texttt{NextNeuron} returns the value of $p(t)$, which is stored as \texttt{next\_potential}.
The value \texttt{next\_potential} is also used to calculate the next output value, which then appears at the head in the new value of \texttt{Output}, i.e., the first argument to \texttt{MakeNeuron} is \texttt{(NextOutput N next\_potential::Output N)}.
The static part of the neuron is just copied over directly to the new neuron, as \texttt{(Feature n).}
The constraint \texttt{CurPot\_Output} is verified by the definitions of \texttt{NextPotential} and \texttt{NextOutput}. Here, \texttt{eq\_refl} serves as a proof that both sides of an equality are identical.

The \texttt{AfterNstepsNeuron} function repeatedly calls \texttt{NextNeuron} on a list of inputs, processing them all.
The result is a new neuron whose output is increased in length equal to the length of \texttt{inp}, and whose value for \texttt{CurPot} is the last output after processing all inputs.
We assume that the argument \texttt{inp} is in reverse order of time.  
Thus both input and output lists are stored in reverse order in order to, as mentioned, simplify the model and proofs.

Again for clarity, we introduce some notation.  We represent the current potential and output of a neuron \texttt{N} that has received inputs \texttt{inp} in an environment with \texttt{len} sources as $\mathit{CurPot}_N(inp, len)$ and $Output_N(inp, len)$, corresponding to the Rocq expressions \texttt{(CurPot (AfterNstepsNeuron N inp len))} and \texttt{(Output (AfterNstepsNeuron N inp len))}, respectively. 
Additionally, we introduce notation for comparison operators whose result is in type $\mathit{bool}$; such operators have $\mathop{?}$ as a subscript.

\begin{figure}[htb]
\vspace{-2.0cm}
\includegraphics[scale=.5,angle=90]{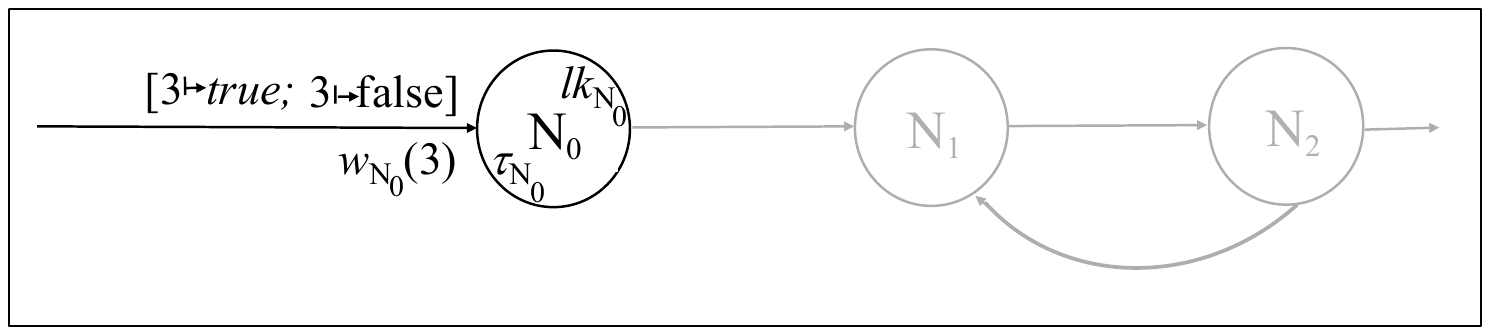}
\vspace{-6.3cm}
\caption{Example with two time steps for $N_0$}
\label{fig:example3}
\end{figure}

\begin{exas}
We again continue Examples~\ref{exas:weights} and~\ref{exas:potential} and illustrate parts of the calculation of the $\mathit{AfterNstepsNeuron}$ function applied to $N_0$.
We consider two time steps starting at time $0$, with an input of $\mathit{true}$ and then $\mathit{false}$, as illustrated in Figure~\ref{fig:example3}.
The inputs come from external source $N_3$, so only argument $3$ of the input functions will matter.  In particular, we have:
$$\begin{array}{rcllll}
w_{N_0} & = & \{0\mapsto 0,&1\mapsto 0,&2\mapsto 0,&3\mapsto w_{N_0}(3)\}\\
\mathit{inp}_1 & = & \{0\mapsto \cdots,&1\mapsto \cdots,&2\mapsto\cdots,&3\mapsto \mathit{true}\}\\
\mathit{inp}_2 & = & \{0\mapsto \cdots,&1\mapsto \cdots,&2\mapsto\cdots,&3\mapsto \mathit{false}\}\\
\end{array}$$
Expanding the call 
$\mathit{AfterNstepsNeuron}(N_1,[\mathit{inp}_2;\mathit{inp}_1],4)$, we get:
$$\mathit{NextNeuron}(\mathit{inp}_2,4,\mathit{NextNeuron}(\mathit{inp}_1,4,N_0))$$
We name the neurons resulting from each of the calls to $\mathit{NextNeuron}$ starting with the innermost call:
$$\begin{array}{rlc}
N_0^1 & = & \mathit{NextNeuron}(\mathit{inp}_1,4,N_0)\\[4pt]
N_0^2 & = & \mathit{NextNeuron}(\mathit{inp}_2,4,N_0^1)\\[4pt]
\end{array}$$
Starting with $N_0^1$, we have:
$$\begin{array}{rcl}
N_0^1 &=&
\mathit{let}~\mathit{np}_1 := \mathit{NextPotential}(N_0,inp_1,4)~\mathit{in}\\
&&\mathit{MakeNeuron}(\mathit{NextOutput}(N_0,\mathit{np}_1)\mathop{::}\mathit{Output}(N_0),\mathit{np}_1,\mathit{Feature}(N_0),\ldots)
\end{array}$$
We elide the last argument of $\mathit{MakeNeuron}$.  Expanding the definition of $\mathit{NextPotential}$, we have:
$$\begin{array}{rcl}
N_0^1 & = & \mathit{let}~\mathit{np}_1 := 
\mathit{if}~\tau_{N_0}\le \mathit{Curpot(N_0)}\\
&& \qquad\qquad\qquad\mathit{then}~\mathit{potential}(w_{N_0},inp_1,4)\\
&& \qquad\qquad\qquad\mathit{else}~\mathit{potential}(w_{N_0},inp_1,4)+\mathit{lk}_{N_0}\cdot\mathit{CurPot}(N_0)~\mathit{in}\\
&&\mathit{MakeNeuron}(\mathit{NextOutput}(N_0,\mathit{np}_1)\mathop{::}\mathit{Output}(N_0),\mathit{np}_1,\mathit{Feature}(N_0),\ldots)
\end{array}$$
Since we start at time 0, $\mathit{Output}(N_0)$ is $[\mathit{false}]$, $\mathit{Curpot}(N_0)$ is 0, and $\tau_{N_0}\le \mathit{Curpot}(N_0)$ is false.  Thus $\mathit{lk}_{N_0}\cdot\mathit{CurPot}(N_0)=0$ and we can simplify:
$$\begin{array}{rcl}
N_0^1 &=&
\mathit{let}~\mathit{np}_1 :=
\mathit{potential}(w_{N_0},inp_1,4)~\mathit{in}\\
&&\mathit{MakeNeuron}([\mathit{NextOutput}(N_0,\mathit{np}_1);0],\mathit{np}_1,\mathit{Feature}(N_0),\ldots)
\end{array}$$
We calculate the potential as we did in Example~\ref{exas:potential}, and obtain $\mathit{potential}(w_{N_0},inp_1,4)=w_{N_0}(3)$.
Note that $\mathit{NextOutput}(N_0,\mathit{np}_1) = \tau_{N_0} \le_? w_{N_0}(3)$. Thus we can further simplify and obtain:
$$N_0^1 =\mathit{MakeNeuron}([\tau_{N_0} \le_? w_{N_0}(3);0],w_{N_0}(3),\mathit{Feature}(N_0),\ldots)$$
Next, we consider $N_0^2$:
$$\begin{array}{rcl}
N_0^2 & = & \mathit{let}~\mathit{np}_2 := 
\mathit{if}~\tau_{N_0^1}\le \mathit{Curpot(N_0^1)}\\
&& \qquad\qquad\qquad\mathit{then}~\mathit{potential}(w_{N_0^1},inp_2,4)\\
&& \qquad\qquad\qquad\mathit{else}~\mathit{potential}(w_{N_0^1},inp_2,4)+\mathit{lk}_{N_0^1}\cdot\mathit{CurPot}(N_0^1)~\mathit{in}\\
&&\mathit{MakeNeuron}(\mathit{NextOutput}(N_0^1,\mathit{np}_2)\mathop{::}\mathit{Output}(N_0^1),\mathit{np}_2,\mathit{Feature}(N_0^1),\ldots)
\end{array}$$
Note that $\mathit{Curpot(N_0^1)}=w_{N_0}(3)$ and $\mathit{Output}(N_0^1)=[\tau_{N_0} \le_? w_{N_0}(3);0]$.
Also, $\mathit{Feature}(N_0)=\mathit{Feature}(N_0^1)$ and thus $\tau_{N_0^1}=\tau_{N_0}$, $w_{N_0^1}=w_{N_0}$, and $\mathit{lk}_{N_0^1}=\mathit{lk}_{N_0}$.  
In addition, since $\mathit{inp}_2(3)=\mathit{false}$, we have $\mathit{potential}(w_{N_0},inp_2,4)=0$.  
Also $\mathit{NextOutput}(N_0^1,\mathit{np}_2) = \tau_{N_0} \le_? \mathit{np}_2$.
Simplifying, we get:
$$\begin{array}{rcl}
N_0^2 & = & \mathit{let}~\mathit{np}_2 := 
\mathit{if}~\tau_{N_0}\le w_{N_0}(3)~\mathit{then}~0~\mathit{else}~\mathit{lk}_{N_0}\cdot w_{N_0}(3)~\mathit{in}\\
&&\mathit{MakeNeuron}([\tau_{N_0} \le_? \mathit{np}_2;\tau_{N_0} \le_? w_{N_0}(3);0],\mathit{np}_2,\mathit{Feature}(N_0),\ldots)
\end{array}$$
Note that if $\mathit{np}_2=0$, then $\tau_{N_0} \le_? \mathit{np}_2$ is false.  Thus we have:
$$N_0^2 = 
\begin{cases} 
\mathit{MakeNeuron}([0;1;0],0,\mathit{Feature}(N_0),\ldots)\qquad & \mathit{if\ } \tau_{N_0}\le w_{N_0}(3), \\
\mathit{MakeNeuron}([0;0;0], &\mathit{otherwise}.\\
\qquad\qquad\qquad\mathit{lk}_{N_0}\cdot w_{N_0}(3),\mathit{Feature}(N_0),\ldots)&
\end{cases}
$$
\end{exas}

The following lemma captures some of the behaviour described above of the functions \texttt{NextNeuron} and \texttt{AfterNstepsNeuron} in Figure~\ref{fig:updateneuron}.
\begin{lem}[Properties of \texttt{AfterNstepsNeuron}]
~~\\[-5pt]
\begin{enumerate}
\setlength{\itemsep}{5pt}
\item Link between \textit{CurPot} and \textit{Output} after $n$ steps:\\
$\begin{array}{l}
\forall (N:\mathit{Neuron}) (i:\mathit{nat}\rightarrow\mathit{bool}) (\mathit{inp}:\mathit{list}~(\mathit{nat}\rightarrow\mathit{bool}))(\mathit{len}:\mathit{nat}),\\
  \mbox{If}~\tau_N \leq_? \mathit{CurPot}_N(\mathit{inp},\mathit{len})\\
  \quad\mbox{then}~\mathit{CurPot}_N(i::\mathit{inp},\mathit{len})\equiv
  \mathit{potential}(\mathit{w}_N,i,\mathit{len})\\
  \quad\mbox{else}~\mathit{CurPot}_N(i::\mathit{inp},\mathit{len})K\equiv
  \mathit{potential}(\mathit{w}_N,i,\mathit{len}) +
  \mathit{lk}_N \cdot \mathit{CurPot}_N(\mathit{inp},\mathit{len})
  \end{array}$\\
\texttt{[AfterNstepsNeuron\_curpot\_cons]}
\item Format of \textit{Output} after $n$ steps:\\
$\begin{array}{l}
\forall (N:\mathit{Neuron}) (i:\mathit{nat}\rightarrow\mathit{bool}) (\mathit{inp}:\mathit{list}~(\mathit{nat}\rightarrow\mathit{bool}))(\mathit{len}:\mathit{nat}),\\
  \mathit{Output}_N(i::\mathit{inp},\mathit{len})=
  (\tau_N \leq_? \mathit{CurPot}_N(i::\mathit{inp},\mathit{len}))::
  \mathit{Output}_N(\mathit{inp},\mathit{len})
  \end{array}$\\
\texttt{[AfterNSN\_curpot\_output\_unfold]}
\end{enumerate}
\label{lem:nsteps-props}
\end{lem}
\noindent
The Rocq statements of these theorems use \texttt{=} for syntactic equality in Rocq and \texttt{==} for equivalence between rational numbers, which is defined in Rocq's rational number library.
When presenting properties, we write $\mathop{=}$ and $\mathop{\equiv}$, respectively, to represent these Rocq operators.

Figure~\ref{fig:wf-equiv} contains definitions expressing an equivalence relation on neurons.
\begin{figure}[htb]
\begin{verbatim}
Definition EquivFeature (nf1 nf2: NeuronFeature) len : Prop :=
  Id nf1 = Id nf2 /\
    (forall id, id < len -> Weights nf1 id == Weights nf2 id) /\
    Leak_Factor nf1 == Leak_Factor nf2 /\
    Tau nf1 == Tau nf2.

Definition EquivNeuron (n1 n2: Neuron) len : Prop :=
  EquivFeature (Feature n1) (Feature n2) len /\
    Output n1 = Output n2 /\ CurPot n1 == CurPot n2.
\end{verbatim}
\caption{Equivalence of neurons}
\label{fig:wf-equiv}
\end{figure}
Two neurons are \textit{equivalent} if all of the values of the static and dynamic fields are the same.
This definition is needed because of the use of rational numbers in the definition of a neuron.  It extends the equivalence defined in Rocq's rational number library to equivalence on neurons.
Note that the definitions do not mention the fields of \texttt{NeuronFeature} and \texttt{Neuron} that represent constraints;
proofs of the constraints do not have to be the same.
In these definitions, the third argument, \texttt{len}, is again the number of neurons in the environment.
We abbreviate \texttt{(EquivNeuron n1 n2 len)} as $n_1\equiv_{len} n_2$.

Lemma~\ref{lem:equivprops} below states some important properties of neuron equivalence, including the preservation of equivalence under the \texttt{NextNeuron} and \texttt{AfterNstepsNeuron} operations.
\begin{lem}[Properties of $\mathop{\equiv_{len}}$]
~~\\[-5pt]
\begin{enumerate}
\setlength{\itemsep}{5pt}
\item The $\mathop{\equiv_{len}}$ relation is reflexive, symmetric, and transitive.\\
\texttt{[EquivNeuron\_refl,EquivNeuron\_sym,EquivNeuron\_trans]}
\item \textit{NextNeuron} preserves equivalence:
  $\forall (N_1~N_2:\mathit{Neuron})  (\mathit{inp}:\mathit{nat}\rightarrow\mathit{bool}) (\mathit{len}:\mathit{nat}),\\
  N_1\equiv_{len} N_2\rightarrow \mathit{NextNeuron}(\mathit{inp},\mathit{len},N_1)\equiv_{len} \mathit{NextNeuron}(\mathit{inp},\mathit{len},N_2)$.\\
\texttt{[nextneuron\_equiv]}
\item \textit{AfterNstepsNeuron} preserves equivalence:\\
  $\forall (N_1~N_2:\mathit{Neuron})  (\mathit{In}:\mathit{list}~(\mathit{nat}\rightarrow\mathit{bool})) (\mathit{len}:\mathit{nat}),\\
  N_1\equiv_{len} N_2\rightarrow \mathit{AfterNstepsNeuron}(N_1,\mathit{In},\mathit{len})\equiv_{len} \mathit{AfterNstepsNeuron}(N_2,\mathit{In},\mathit{len})$.\\
\texttt{[afternstepsneuron\_equiv]}
\end{enumerate}
\label{lem:equivprops}
\end{lem}

\subsection{Defining Circuits and their Properties in Rocq}
\label{subsec:circuits-props}

In our previous work, we encoded each archetype directly as a Rocq record.  
Here, we generalize our work, and include a general definition of neuronal circuits, which we later specialize to define archetypes.
The Rocq record that models a general circuit is shown in Figure~\ref{fig:circuit}.
\begin{figure}[htb]
\begin{verbatim}
Record NeuroCircuit :=
  MakeCircuit {
      Time : nat;
      ListNeuro : list Neuron;
      SupplInput : nat;
      IdNeuroDiff : forall n m l1 l2 l3,
        ListNeuro = l1 ++ n :: l2 ++ m :: l3 ->
        Id (Feature n) <> Id (Feature m);        
      IdInfLen : forall n,
        In n ListNeuro -> Id (Feature n) < length ListNeuro;        
      TimeNeuro : forall n,
        In n ListNeuro -> length (Output n) = Time + 1; 
    }.

Definition is_initial_neuro n len := 
  exists (nf: NeuronFeature), EquivNeuron n (SetNeuron nf) len.

Definition is_initial (nc : NeuroCircuit) := 
 forall n, In n (ListNeuro nc) -> 
   exists nf, 
     EquivNeuron n (SetNeuron nf) (length (ListNeuro nc) + SupplInput nc).
\end{verbatim}
\caption{Rocq definition of general and initial neuronal circuits}
\label{fig:circuit}
\end{figure}
This record contains three data fields and three constraints.  
At this level of generality, the main field is the list of neurons contained in the circuit (\texttt{ListNeuro}).  
We also record the current time (\texttt{Time}).
As mentioned earlier, the input(s) to some neurons in a circuit will be the output(s) of other neurons in the circuit.  
The \texttt{SupplInput} field records the number of external sources of input to the circuit.
The first constraint (\texttt{IdNeuroDiff}) ensures that all the identifiers of the neurons in the circuit are distinct.  
The second (\texttt{IdInfLen}) states that the identifiers appear in a sequence from $0$ up to the number of neurons minus one.  
The last constraint (\texttt{TimeNeuro}) requires that the output lists of each of the neurons in the circuit have the same length, which is greater than the value of \texttt{Time}.  
The extra value will always be an output of \texttt{false} at time $0$, as mentioned earlier.

Figure~\ref{fig:circuit} also contains definitions describing initial neurons (\texttt{is\_initial\_neuro}) and circuits (\texttt{is\_initial}).  
A neuron is \emph{initial} if it is equivalent to a neuron created by the \texttt{SetNeuron} function in Figure~\ref{fig:Coqneuron}. 
As mentioned earlier, such a neuron is at a stage that corresponds to a neuron at time 0.
A circuit is \emph{initial} if all neurons in the circuit are initial.  
Since equivalence depends on the number of neurons in the environment, these definitions do also.
Note that in the body of \texttt{is\_initial}, this is expressed as the sum of the number of neurons in the circuit (the \texttt{ListNeuro} field) and the number of external sources of input (the \texttt{SupplInput} field).

If \texttt{NC} is a \texttt{NeuroCircuit}, we write $t_\mathit{NC}$, $\mathit{ln}_\mathit{NC}$, and $\mathit{si}_\mathit{NC}$ to represent, respectively, \texttt{(Time NC)}, \texttt{(ListNeuro NC)}, and \texttt{(SupplInput NC)}.
In addition, we write $\mathit{is\_initial}_\mathit{Neur}(N, \mathit{len})$ to denote \texttt{(is\_initial\_neuro N len)} and $\mathit{is\_initial}_\mathit{Cir}(\mathit{NC})$ to denote \texttt{(is\_initial NC)}.

We extend the notion of well-formed neurons to circuits.  A circuit is \emph{well-formed}
if all of the neurons in the circuit are well-formed neurons.   In particular, all of the neurons in field \texttt{ListNeuro} must be well-formed.\footnote{See definition \texttt{well\_formed\_circuit} in the Rocq code.}

Below are some important properties that follow from the definitions in Figure~\ref{fig:circuit}.  We abbreviate \texttt{(In n L)} for list membership as $n\in L$.  In addition, here and in the rest of the paper, the booleans are represented by $0$ and $1$ to simplify the reading.
We also sometimes refer to the $0$ value as \emph{null}.
\begin{lem}[Properties of Circuits and of Initial Neurons and Circuits]
~~\\[-5pt]
\begin{enumerate}
\item $\forall (\mathit{id}:\mathit{nat})(\mathit{NC}:\mathit{NeuroCircuit}),\\
\mathit{id}<\mathit{length}(\mathit{ln}_\mathit{NC})\rightarrow
\exists (N:\mathit{Neuron}),N\in\mathit{ln}_\mathit{NC}\land\mathit{id}_N = \mathit{id}.$\\
\texttt{[len\_inf\_in\_listneuro]}
\item $\forall (N_1~N_2:\mathit{Neuron})(\mathit{NC}:\mathit{NeuroCircuit}),\\
N_1,N_2\in \mathit{ln}_\mathit{NC}\rightarrow
\mathit{id}_{N_1} = \mathit{id}_{N_2}\rightarrow N_1 \equiv_{length(ln_{NC})+ si_{NC}}
N_2.$\\
\texttt{[same\_id\_same\_neuron]}
\item $\forall (N:\mathit{Neuron}) (\mathit{len}:\mathit{nat}),\\\quad\mathit{is\_initial}_\mathit{Neur}(N,\mathit{len})\rightarrow \mathit{CurPot}(N)\equiv 0\land \mathit{Output}(N) = [0]$\\
\texttt{[is\_initial\_neuro\_curpot, is\_initial\_neuro\_output]}
\item $\forall (\mathit{NC}:\mathit{NeuroCircuit}),\;
\mathit{is\_initial}_\mathit{Cir}(\mathit{NC})\rightarrow
0 < \mathit{length}(\mathit{ln}_\mathit{NC})\rightarrow
\mathit{t}_\mathit{NC} = 0.$\\
\texttt{[is\_initial\_time]}
\end{enumerate}
\label{lem:circprops}
\end{lem}
\noindent
In Lemma~\ref{lem:circprops}, (1) states that if circuit has $n$ neurons, every number in the range $0,\ldots,n-1$ is an identifier of one of the $n$ neurons;
(2) states that all neurons in a circuit having the same identifier are equivalent; 
(3) states that an initial neuron has a value of $0$ for the current potential, and an output list containing only $0$ (i.e., $\mathtt{false}$), which is the output value at time $0$;
and (4) states that in an initial circuit containing at least one neuron, the \texttt{Time} field must have value $0$.

We next consider Rocq definitions of useful functions on circuits.
\begin{figure}[htb]
\begin{verbatim}
Definition output_neuron_init nc id :=
  match find (fun x => Id (Feature x) =? id) (ListNeuro nc) with
  |Some x => Output x
  |None => []
  end.

Definition NextNeuroInNC (nc : Neurocircuit) (inp : nat -> bool) :=
  NextNeuron (fun x => if (x <? length (ListNeuro nc))
                       then hd false (output_neuron_init nc x)
                       else inp x)
             (length (ListNeuro nc) + SupplInput nc).
             
Definition NextListNeuro (nc : Neurocircuit) (inp : nat -> bool) :=
  map (NextNeuroInNC nc inp) (ListNeuro nc).
  
Definition NextStepC (nc : Neurocircuit) inp :=
  let listneuro := NextListNeuro nc inp in
  Makecircuit (Time nc + 1) listneuro (SupplInput nc)
              (id_neuro_diff_next nc inp)
              (id_inf_numb_neuron_next nc inp)
              (next_time _ _).
              
Fixpoint Nstepscircuit (nc : Neurocircuit) (inp : list (nat -> bool)) :
  Neurocircuit :=
  match inp with
  | nil => nc
  | h::tl => (NextStepC (Nstepscircuit nc tl) h) 
  end.
\end{verbatim}
\caption{Processing input to circuits}
\label{fig:circuitinputs}
\end{figure}
The first definition in Figure~\ref{fig:circuitinputs} returns the output of a neuron with identifier \texttt{id} in a circuit \texttt{nc} if there is such a neuron; otherwise an empty list (no output) is returned.
It is used in the next function and only applied to identifiers of neurons that belong to the circuit.
In \texttt{NextNeuroInNC}, the second argument (\texttt{inp}) is an input function where only the values of the identifiers of neurons serving as external sources of inputs matter.
The body of this function is a call to \texttt{NextNeuron} where only two of the three arguments are given.
The first argument to \texttt{NextNeuron} is an input function that takes as argument the identifier of a neuron, and if the neuron is in the circuit, it looks up the last value of its output, which will serve as input to the neurons it is connected to.
If the neuron is an external source, then the input function \texttt{inp} is used to determine the input value.  The second argument is the number of neurons in the environment, which is the number of neurons in the circuit plus the number of external sources.
A call to \texttt{NextNeuroInNC} creates a function of type \texttt{Neuron -> Neuron} such that when it is applied to a neuron, will create a new neuron with one more output value and an updated value of \texttt{CurPot}.

The next function in Figure~\ref{fig:circuitinputs} (\texttt{NextStepC}), applies \texttt{NextNeuroInNC} to all the neurons in a circuit, creating a new list of neurons that have all been updated by processing input at one more time step.  
It creates a new circuit with this new list of neurons, the time increased by $1$, and the value of the \texttt{SupplInput} field unchanged.
The last three arguments to \texttt{Makecircuit} in the body of \texttt{NextStepC} are calls to functions which update the proofs of the three constraints.
The definitions of these functions are omitted.

The last function in Figure~\ref{fig:circuitinputs} (\texttt{Nstepscircuit}) takes a list of input functions of some length $n$, and updates a circuit one step at a time, increasing the time and the length of the output lists of each of the neurons by $n$.

Figure~\ref{fig:circuitfuns} contains two simple functions extracting information from a circuit after it has been updated by processing a list of inputs using \texttt{Nstepscircuit}.
\begin{figure}[htb]
\begin{verbatim}
Definition output_neuron nc inp id :=
  match find (fun x => Id (Feature x) =? id)
             (ListNeuro (Nstepscircuit nc inp)) with
  |Some x => Output x
  |None => []
  end.

Definition curpot_neuron nc inp id :=
  match find (fun x => Id (Feature x) =? id)
            (ListNeuro (Nstepscircuit nc inp)) with
  |Some x => CurPot x
  |None => 0
  end.
\end{verbatim}
\caption{Functions extracting information from circuits}
\label{fig:circuitfuns}
\end{figure}
They both take a circuit, a list of inputs, and a neuron identifier as arguments.
After processing the list of inputs, the first function returns the new output of the neuron with the given identifier, if there is one; otherwise it returns an empty list.
The second function is similar, but returns the new value of \texttt{CurPot}.
The following lemma expresses that the output and current potential, respectively, of each neuron in a circuit is unchanged in the case when the second argument to these functions (the input list) is empty, i.e., the output and potential do not change after $0$ steps.
If \texttt{NC} is a \texttt{NeuroCircuit}, \texttt{N} is a neuron of \texttt{NC} with identifier \texttt{id}, and \texttt{inp} is a list of input functions, we write $\mathit{output}_{\mathit{NC}}(N, \mathit{inp})$ and $\mathit{curpot}_{\mathit{NC}}(N, \mathit{inp})$ to represent, respectively, \texttt{(output\_neuron NC inp id)} and \texttt{(output\_neuron NC inp id}).
\begin{lem}[Output and CurPot Unchanged with Empty Input]
~~\\[-5pt]
\begin{enumerate}
\item $\forall (N:\mathit{Neuron}) (\mathit{NC}:\mathit{NeuroCircuit}),
N\in \mathit{ln}_\mathit{NC}\rightarrow
\mathit{output}_\mathit{NC}(N,[\ ])=\mathit{Output}(N)$\\
\texttt{[output\_neuron\_Output]}
\item $\forall (N:\mathit{Neuron}) (\mathit{NC}:\mathit{NeuroCircuit}),
N\in \mathit{ln}_\mathit{NC}\rightarrow
\mathit{curpot}_\mathit{NC}(N,[\ ])=\mathit{CurPot}(N)$\\
\texttt{[curpot\_neuron\_CurPot]}
\end{enumerate}
\label{lem:nstepscorrect}
\end{lem}

\subsection{Defining Archetypes and their Properties in Rocq}
\label{subsec:archetypes-props}

Figure~\ref{fig:arch-series} contains the Rocq definition of the simple series archetype, as seen in Figure~\ref{fig:archetypes}(a).
\begin{figure}[htb]
\begin{verbatim}
Record Series (c : NeuroCircuit) :=
  Make_Series
    {
      OneSupplementS : SupplInput c = 1;
      ExistsNeuronS : (0 < length (ListNeuro c));
      FirstNeuronS : forall n, 
        In n (ListNeuro c) -> Id (Feature n) = 0 ->
        (forall id' : nat, (id' < length (ListNeuro c)) ->
          Weights (Feature n) id' == 0) /\
          0 < Weights (Feature n) (length (ListNeuro c));
      UniqueWeightS : forall n m, 
        In n (ListNeuro c) -> Id (Feature n) = S m ->
        (forall id' : nat, m <> id' -> (id' < length (ListNeuro c) + 1) ->
          Weights (Feature n) id' == 0) /\
        0 < Weights (Feature n) m
    }.
\end{verbatim}
\caption{Rocq representation of archetype (a) in Figure~\ref{fig:archetypes}}
\label{fig:arch-series}  
\end{figure}
In general, our approach is to encode the particular structure of each archetype as a Rocq record.  
Records can have input parameters, similar to functions in Rocq, and here the input parameter is a \texttt{NeuroCircuit}, as defined in Figure~\ref{fig:circuit}.
Here, \texttt{Series} is actually a predicate taking one argument of type \texttt{NeuroCircuit}.
All fields in this record are constraints that must hold of the circuit in order to give it the structure of a simple series.  
In particular, if \texttt{c} is a \texttt{NeuroCircuit}, then \texttt{(Series n)} expresses that \texttt{c} is a circuit satisfying all of these constraints.
All other Rocq definitions of archetypes will have a similar structure, i.e., have a \texttt{NeuroCircuit} parameter and define structural constraints.

The first constraint (\texttt{OneSupplementS}) states that there is exactly one external input to the circuit.  
The second constraint (\texttt{ExistsNeuronS}) states that a simple series must have at least one neuron.  
The next two constraints (\texttt{FirstNeuronS} and \texttt{UniqueWeightS}) express properties of the first neuron in the series, and the rest of the neurons in the series, respectively.
The first conjunct of \texttt{FirstNeuronS} states that the weights of all neurons in the list of neurons must be $0$ because none on them have outputs that are connected to the first neuron.  
The second conjunct says that the weight on the external input must be positive.  
This positivity condition holds of all inputs to all neurons in the archetypes of Figure~\ref{fig:archetypes} except those with inhibiting edges, which must have negative weights.
The first conjunct of \texttt{UniqueWeightS} states that given any neuron except the first one, i.e., a neuron whose identifier is $m+1$ for some $m$, all neurons whose identifier is not $m$ has a weight of $0$.
In other words, no other neuron except neuron $m$ has an output that serves as an input to neuron $m+1$.
Note that this conjunct includes the restriction that the external source of input is not connected to any neuron other than the first one.
The second conjunct expresses that the weight on the input to neuron $m+1$ coming from neuron $m$ must be positive.

Note that the definition in Figure~\ref{fig:arch-series} also covers the archetype in Figure~\ref{fig:archetypes}(b) because we can always compute and look up the output of any of the neurons in a circuit.

The Rocq definition of the parallel composition archetype in Figure~\ref{fig:archetypes}(c) appears in Figure~\ref{fig:arch-para-comp}.
\begin{figure}[htb]
\begin{verbatim}
Record ParallelComposition (c : NeuroCircuit) :=
  Make_ParaComp
    {
      OneSupplementPC : ...
      ExistsNeuronPC : ...
      FirstNeuronPC : ...
      N1NmPC : forall n m,
        In n (ListNeuro c) -> Id (Feature n) = S m ->
        (forall id' : nat, 0 <> id' -> id' < length (ListNeuro c) + 1 ->
          Weights (Feature n) id' == 0) /\
        0 < Weights (Feature n) 0
    }.
\end{verbatim}
\caption{Rocq representation of archetype (c) in Figure~\ref{fig:archetypes}}
\label{fig:arch-para-comp}  
\end{figure}
We elide the definitions of \texttt{OneSupplementPC}, \texttt{ExistsNeuronPC}, and \texttt{FirstNeuronPC} because they are exactly the same as the definitions of \texttt{OneSupplementS}, \texttt{ExistsNeuronS}, and \texttt{FirstNeuronS} in Figure~\ref{fig:arch-series}, respectively.
The last constraint is also similar, except that two occurrences of \texttt{m} in \texttt{UniqueWeightS} are replaced by \texttt{0} in \texttt{N1NmPC}.
The first conjunct is thus about neurons other than the one with identifier $0$, in particular, all neurons with identifiers in the range $1,\ldots,i$, where $i$ is the number of identifiers in the environment.
This includes the external source of output, which has identifier $i$, plus all the neurons that occur in parallel, which have identifiers $1,\ldots,i-1$.  
For every identifier $\mathit{id}'$ in this range, we have $w_n(\mathit{id}')=0$.
Thus, the neurons that occur in parallel have no connections between them, and the external output also has no connections to any of them.
The second conjunct is about neuron $0$, which is connected to all other neurons in the circuit.  These connections are expressed by stating that the weight on the input from $0$ to every other neuron is positive.

The Rocq definition of the positive loop archetype is in Figure~\ref{fig:arch-pos-loop}.
\begin{figure}[htb]
\begin{verbatim}
Record PositiveLoop (c : NeuroCircuit) :=
  Make_PosLoop
    {
      OneSupplementPL : SupplInput c = 1
      ListNeuroLengthPL : length (ListNeuro c) = 2;
      FirstNeuronPL : forall n,
        In n (ListNeuro c) -> Id (Feature n) = 0 ->
        0 < Weights (Feature n) 1 /\ 0 < Weights (Feature n) 2;
      SecondNeuroPL : forall n,
        In n (ListNeuro c) -> Id (Feature n) = 1 ->
        0 < Weights (Feature n) 0 /\ Weights (Feature n) 2 == 0;
    }.
\end{verbatim}
\caption{Rocq representation of archetype (d) in Figure~\ref{fig:archetypes}}
\label{fig:arch-pos-loop}  
\end{figure}
In a positive loop, there is again one external source of input (constraint \texttt{OneSupllementPL}).  
Here the number of neurons in the circuit is exactly $2$ (constraint \texttt{ListNeuroLengthPL}).  
The neuron with identifier $0$ in our Rocq definition of this archetype represents the top neuron in Figure~\ref{def.snn}, which has 2 inputs, an external one, and one coming from the other neuron in the circuit.
In the Rocq definition, this other neuron has identifier $1$.
The neuron with identifier $1$ has one input coming from neuron $0$. 
The identifier of the external input is $2$.
The details of these connections are expressed by the last two constraints (\texttt{FirstNeuronPL} and \texttt{SecondNeuroPL)}, which state that the weights of these three inputs must be positive, and the weight where there is no connection (neuron $2$ to $1$) must be $0$.

The records for the remaining archetypes from Figure~\ref{fig:archetypes} are similar to those presented so far.
We describe them briefly here, and the reader is referred to Appendix~\ref{app:archetypes} for the full details.  
All records defining archetypes have four constraints.
First, consider the negative loop, which is very similar to the positive loop.  
The only difference is that the connection from neuron $1$ to neuron $0$ must have a negative weight.  
In the Rocq definition, 3 of the 4 constraints are exactly the same;
the only difference is in the third constraint.  
The first conjunct in \texttt{FirstNeuronPL} in Figure~\ref{fig:arch-pos-loop} defining the positive loop, which is:
\begin{verbatim}
0 < Weights (Feature n) 1
\end{verbatim}
is replaced by the following in constraint \texttt{FirstNeuronNL} in record \texttt{NegativeLoop}:
\begin{verbatim}
Weights (Feature n) 1 < 0.
\end{verbatim}

We use our mathematical notation when discussing the last two archetypes, which are inhibition and contralateral inhibition (\texttt{Inhibition} and \texttt{ContraInhib} in Rocq); the reader is again referred to the appendix for the Rocq code.
Both of these archetypes have two external sources of input.
Given $c$ of type $\mathit{NeuroCircuit}$, the first constraint is thus $\mathit{SupplInput}(c) = 2$.  
Like the positive and negative loops, these two archetypes have exactly $2$ neurons, so the second constraint does not change.
The two neurons in the circuit again have identifiers $0$ and $1$ for the top and bottom neurons in Figure~\ref{fig:archetypes}, respectively; we refer to them as $N_0$  and $N_1$.
The neurons providing input to $N_0$ and $N_1$ have identifiers $2$ and $3$, respectively, and thus we refer to them as $N_2$ and $N_3$.
The third constraint for these new archetypes again expresses the connections for the top neuron, $N_0$.
For the inhibition archetype, these connections are defined as:
$$w_{N_0}(1) = 0\land 0 < w_{N_0}(2)\land w_{N_0}(3) = 0.$$
The only input to $N_0$ comes from external source $N_2$, and it must be positive.

The connections for $N_1$, expressed as the fourth constraint in the Rocq definition, are:
$$w_{N_1}(0) < 0\land w_{N_1}(2) = 0\land 0 < w_{N_1}(3).$$
The connection coming from $N_0$ must be negative, and the one coming from external source $N_3$ must be positive.

For contralateral inhibition, the last two constraints are:
$$\begin{array}{ccccc}
w_{N_0}(1) < 0 & \land & 0 < w_{N_0}(2) & \land & w_{N_0}(3) = 0\\
w_{N_1}(0) < 0 & \land & w_{N_1}(2) = 0 & \land & 0 < w_{N_1}(3).
\end{array}$$
Compared to inhibition, in contralateral inhibition, $N_0$ has one additional connection coming from $N_1$ and it must be negative.  $N_1$ has the same connections in both of these archetypes, and so the last constraint is the same.

\section{Properties of Neurons and their Proofs}
\label{sec:neuronprops}

In Section~\ref{sec:neuronprops}, we present six properties related to the behaviors of neurons based on their features and the inputs they receive from the environment.
These properties pertain to the neuron's output and current potential behavior when it is in an initialized state and receives a series of inputs from an environment.
At this stage, we do not differentiate between inputs from internal neurons and those from external sources, as the circuit boundaries have not been defined.
In Section~\ref{subsec:multipleinput}, we concentrate on two properties of multiple-input neurons, i.e., neurons receiving inputs from potentially different sources.
Finally, in Section~\ref{subsec:singleinput}, we describe four properties about single-input neurons, which are neurons that receive inputs from at most one source.

\subsection{Properties of Multiple-Input Neurons}
\label{subsec:multipleinput}
The first two properties examine the behavior of a neuron \texttt{N} with potentially multiple sources providing inputs, where all incoming edge weights are either uniformly non-negative or uniformly non-positive. These properties demonstrate the impact of multiple excitatory and inhibitory neurons on the neuron's activity. Excitatory neurons promote the firing of a neuron by increasing the potential, driving it into the non-negative range. In contrast, inhibitory neurons hinder firing by decreasing the potential value.

\subsubsection{Property: Always Non-Negative}
In the case of the neuron having only excitatory neurons transmitting inputs to an initialized neuron, in our model, the current potential (as expressed by the $\mathit{CurPot}_N$) function within the neuron is always non-negatively charged.
\begin{lem}[$\mathit{CurPot}_N$ Always Non-Negative for Multiple-Input Neuron] \texttt{[Always\_N\_Neg]}\\\ \\
$\begin{array}{l}
  \forall (N:\mathit{Neuron})  (\mathit{inp}:\mathit{list} ~ (\mathit{nat}\rightarrow\mathit{bool})) (\mathit{len}:\mathit{nat}),\\
  \quad\mathit{is\_initial}_\mathit{Neur}(\mathit{N, len)} ~ \land ~
  (\forall (\mathit{id}:\mathit{nat}), \mathit{id} < \mathit{len} \rightarrow w_N(\mathit{id})\geq 0)\rightarrow\\
  \quad\quad \mathit{CurPot}_N(\mathit{inp}, \mathit{len})\geq0.
  \end{array}$
\label{lem:nonnegcurr}
\end{lem}

\begin{proof}
The proof proceeds by induction on the structure of the input sequence.\\
\textbf{Base case}: $\mathit{inp}=[\ ]$ (the empty list).

When there is no input in the input sequence, the current potential of a neuron is the same as the one at time 0, i.e., $\mathit{CurPot}_N([\ ], len) = \mathit{CurPot}(N)$.
By the properties of $\mathit{is\_initial}_\mathit{Neur}$ in Lemma~\ref{lem:circprops} (3): $CurPot(N) \equiv 0$. Hence, $CurPot_N([\ ], len)$ is non-negative.\\
\textbf{Induction case}: We assume that the property holds for the input sequence $\mathit{inp}$, and we must show that it also holds for $\mathit{i::inp}$, where $i$ is an additional input value.\\
With Lemma~\ref{lem:nsteps-props} (1), we unfold the function $\mathit{CurPot}_N$:
\[
\mathit{CurPot}_N(i::\mathit{inp}, \mathit{len}) \equiv 
\begin{cases}
\mathit{potential}(w_N, \mathit{inp}, \mathit{len}) & \mathit{if}~ \tau_N \leq \mathit{CurPot}_N(\mathit{inp}, \mathit{len}), \\
\mathit{potential}(w_N, \mathit{inp}, \mathit{len})~\mathop{+} & \mathit{otherwise}.\\
\quad \mathit{lk}_N \cdot \mathit{CurPot}_N(\mathit{inp}, \mathit{len}) & \\
\end{cases}
\]
We divide the proof in the two sub-cases: one where the neuron $\mathit{N}$ fires after processing the input sequence $\mathit{inp}$ or not.

\textbf{Sub-case $\tau_N \leq \mathit{CurPot}_N(\mathit{inp},\mathit{len})$} (firing). 
By Lemma~\ref{lem:nonnegpot}, $\mathit{potential}(w_N, \mathit{inp}, \mathit{len})\geq 0$, since the weights are non-negative. We conclude that $\mathit{CurPot}_N(i::\mathit{inp},\mathit{len})$ is non-negative.

\textbf{Sub-case $\tau_N > \mathit{CurPot}_N(\mathit{inp},\mathit{len})$} (no firing).
With the same reasoning as the previous case, $\mathit{potential}(w_N, \mathit{inp}, \mathit{len})$ is non-negative. By the induction hypothesis, we have $\mathit{CurPot}_N(\mathit{inp},len) \geq 0$ and by definition, $\mathit{lk}_N\geq 0$.
We can conclude that $\mathit{CurPot_N}(i::\mathit{inp}, \mathit{len}) \geq 0$.
\end{proof}

\subsubsection{Property: Inhibitor Effect}
A initialized neuron with inputs coming exclusively from inhibitory neurons never fires, i.e., the neuron is blocked in the state of non-firing.
\begin{lem}[Inhibitor Effect for a Multiple-Input Neuron] \texttt{[Inhibitor\_Effect]}\\\ \\
$\begin{array}{l}
  \forall (N:\mathit{Neuron})  (\mathit{inp}:\mathit{list}~ (\mathit{nat}\rightarrow\mathit{bool})) (\mathit{len}:\mathit{nat}),\\
  \quad \mathit{is\_initial}_\mathit{Neur}(N, \mathit{len}) ~\land ~
  (\forall (\mathit{id}:\mathit{nat}), \mathit{id} < \mathit{len} \rightarrow w_N(\mathit{id})\leq 0)\rightarrow\\
  \quad\quad\forall (\mathit{a}:\mathit{bool}), a \in \mathit{Output}_N(\mathit{inp}, \mathit{len})\rightarrow a = 0.
  \end{array}$
\label{lem:inhibit}
\end{lem}
\begin{proof}
The proof is similar to the one of Lemma~\ref{lem:nonnegcurr}. It proceeds by induction on the structure of the input sequence.\\
\textbf{Base case}: $\mathit{inp}=[\ ]$ (the empty list).

Since there is no input in the input sequence and by the properties of $\mathit{is\_initial}_\mathit{Neur}$ (Lemma~\ref{lem:circprops}): $Output_N([\ ], len) = Output(N) = [0]$. The property is verified for the base case.\\
\textbf{Induction case}: We assume that the property holds for the input sequence $\mathit{inp}$, and we must show that it also holds for $\mathit{i::inp}$, where $i$ is an additional input value.\\
By Lemma~\ref{lem:nsteps-props}(2): 
\[
\mathit{Output}_N(i::\mathit{inp}, \mathit{len}) = 
\begin{cases}
1 \mathop{::} \mathit{Output}_N (\mathit{inp}, \mathit{len})& \mathit{if}~ \tau_N \leq \mathit{CurPot}_N(i::\mathit{inp}, \mathit{len}), \\
0 \mathop{::} \mathit{Output}_N (\mathit{inp}, \mathit{len}) & \mathit{otherwise}.\\
\end{cases}
\]
By the induction hypothesis, $\forall a \in \mathit{Output}_N (\mathit{inp}, \mathit{len}) \rightarrow a = 0$.
We remark that the property is verified if and only if $\tau_N > \mathit{CurPot}_N(i::\mathit{inp}, \mathit{len})$.\\
By Lemma~\ref{lem:nsteps-props} (1): 
\[
\mathit{CurPot}_N(i::\mathit{inp}, \mathit{len}) \equiv 
\begin{cases}
\mathit{potential}(w_N, \mathit{inp}, \mathit{len}) & \mathit{if}~ \tau_N \leq \mathit{CurPot}_N(\mathit{inp}, \mathit{len}), \\
\mathit{potential}(w_N, \mathit{inp}, \mathit{len})~\mathop{+} & \mathit{otherwise}.\\
\quad \mathit{lk}_N \cdot \mathit{CurPot}_N(\mathit{inp}, \mathit{len}) & \\
\end{cases}
\]

We divide the proof in the two sub-cases to show that we always have: $\tau_N >\mathit{CurPot}_N(i::\mathit{inp},\mathit{len})$. 

\textbf{Sub-case $\tau_N \leq \mathit{CurPot}_N(\mathit{inp},\mathit{len})$} (firing). 
By Lemma~\ref{lem:nonpospot}, $\mathit{potential}(w_N, \mathit{inp}, \mathit{len})\leq 0$ since the weights are non-positive. Hence, $\mathit{CurPot}_N(i::\mathit{inp},\mathit{len})$ is non-positive. By design of a neuronal feature, $\tau_N>0$. We conclude that $\mathit{CurPot}_N(i::\mathit{inp},\mathit{len})<\tau_N$.

\textbf{Sub-case $\tau_N > \mathit{CurPot}_N(\mathit{inp},\mathit{len})$} (no firing).
Using the same reasoning as the previous case, $\mathit{potential}(w_N, \mathit{inp}, \mathit{len}) \leq 0$. By definition of the leak factor, we have $0 < lk_N < 1$, and by the hypothesis of this subcase, we have $\mathit{CurPot}_N(\mathit{inp},\mathit{len}) < \tau_N$. 
Hence, $lk_N\cdot \mathit{CurPot}_N(\mathit{inp},\mathit{len}) < \tau_N$.
We can conclude that $\mathit{CurPot}_N(i::\mathit{inp}, \mathit{len}) < \tau_N$.
\end{proof}

\subsection{Properties of Single-Input Neurons}
\label{subsec:singleinput}
In this subsection, we analyze the behavior of a neuron with at most a single source of input, starting in its initialized state. We examine cases where the source's weight is either greater or less than the activation threshold. We also study the effect of having $1$s in the input sequences on the neuron’s output, particularly on the number of firings.

In Figure \ref{fig:oneinput}, we introduce the definition of \texttt{One\_input\_Or\_Less} to characterize single-input neurons.
\begin{figure}[htb]
\begin{verbatim}
Definition One_Input_Or_Less (nf : NeuronFeature) (id len : nat) :=  
  id < len /\
  (forall id', id <> id' -> id' < len -> Weights nf id' == 0).
\end{verbatim}
\caption{Rocq definitions for one-input neurons}
\label{fig:oneinput}  
\end{figure}
The property takes three arguments: \texttt{nf}, represents the neuronal feature of the neuron of interest, \texttt{id} is the identifier of a potential input source for that neuron, and \texttt{len} corresponds to the number of neurons in the environment.
Recall that each element in the environment is numbered sequentially, starting from 0 with a step of 1.
To ensure that the identifier \texttt{id} is within the environment, the first part of the definition requires \texttt{id} to be strictly less than the number of elements in the environment \texttt{len}.
At this stage, no restrictions are imposed on the neuron's identifier itself, contrary to a neuron in a circuit.
The second part of the property verifies that every neuron in the environment other than the one with the identifier \texttt{id} is not a source of input to the neuron of interest with feature \texttt{nf}, i.e., the value of \texttt{(Weights nf)} is null for every identifier other than \texttt{id}.

The notation we adopt for this single-input property is $\mathit{One\_input}(N, \mathit{id}, \mathit{len})$ where the neuron is $N$,  its one source of input has identifier $\mathit{id}$, and it is in an environment of $\mathit{len}$ neurons.

Before we introduce the properties on the behaviours of the current potential or the output of a single neuron (as expressed by the $\mathit{CurPot}_N$ and $\mathit{Output}_N$ functions, respectively), we first consider some properties about the \texttt{potential} function from Figure~\ref{fig:Coqpotential}; these properties are necessary to understand the proofs of our properties of interest.

The value of the potential function is reduced to two simple cases depending only on the value of the weight of the single source of input, here $\mathit{id}$, and the value of the input from $\mathit{id}$. 
\begin{lem}[Simplifed $\mathit{potential}$ Function for a Single-Input Neuron]\texttt{[potential\_oi]}\\\ \\
$\begin{array}{l}
  \forall (N:\mathit{Neuron}) (\mathit{id}:\mathit{nat}) (\mathit{inp}:\mathit{nat}\rightarrow\mathit{bool}) (\mathit{len}:\mathit{nat}),\\
  \quad \mathit{One\_input}(N, \mathit{id}, \mathit{len})\rightarrow\\
   \quad\mathit{potential}(w_N, \mathit{inp}, \mathit{len})\equiv\begin{cases}
  w_N (\mathit{id}) & \mathit{if}~\mathit({inp}~\mathit{id}) = \mathit{true}, \\
  0 & \mathit{otherwise}.\\
\end{cases}
  \end{array}$
\label{lem:potoi}
\end{lem}
\noindent
Every weight other than the one for $\mathit{id}$ is null, leaving only $w_N(id)$ to be potentially non-null.

The properties below follow from Lemma~\ref{lem:potoi}.
\begin{lem}[Properties of the $\mathit{potential}$ Function for Single-Input Neurons]
~~\\[-5pt]
\begin{enumerate}
\item $\forall (N:\mathit{Neuron}) (\mathit{id}:\mathit{nat}) (\mathit{inp}:\mathit{nat}\rightarrow\mathit{bool}) (\mathit{len}:\mathit{nat}) (\mathit{m}:\mathit{nat}),\\
  \quad \mathit{One\_input}(N, \mathit{id}, \mathit{len})\land 0 < m\land m\leq w_N(\mathit{id})\rightarrow\\
   \quad\quad (m \leq_? \mathit{potential}(w_N, \mathit{inp}, \mathit{len})) = \mathit{inp}~\mathit{id}$.\\ \texttt{[potential1N\_w\_greater\_n]}
\item $\forall (N:\mathit{Neuron}) (\mathit{id}:\mathit{nat}) (\mathit{inp}:\mathit{nat}\rightarrow\mathit{bool}) (\mathit{len}:\mathit{nat}) (\mathit{m}:\mathit{nat}),\\
  \quad \mathit{One\_input}(N, \mathit{id}, \mathit{len})\land 0 < m\land m > w_N(\mathit{id})\rightarrow\\
   \quad\quad m > \mathit{potential}(w_N, \mathit{inp}, \mathit{len})$.\\ \texttt{[potential1N\_w\_less\_n]}
\end{enumerate}
\label{lem:sineurprops}
\end{lem}
\noindent
In the case where the weight of the identifier $\mathit{id}$ is greater than or equal to a positive integer $m$, the potential is only greater than $m$ if the input for $\mathit{id}$ is $1$.
If the weight for $\mathit{id}$ is smaller than $m$, the potential is always smaller than $m$.

Since there is only one input at each time step in the single-input case, we simplify the type for the inputs and write $\mathit{list}~\mathit{bool}$ instead of $\mathit{list}~(\mathit{nat} \rightarrow \mathit{bool})$ moving forward.

\begin{rem}
In the case of a single-input neuron, the current potential and the output functions ($\mathit{CurPot}_N$ and $\mathit{Output}_N$, respectively) can be further simplified when the weight for single-input source reaches the threshold.

Let's start with the definition of the current potential given by Lemma \ref{lem:nsteps-props} (1) for an input sequence $i::\mathit{inp}$.  Since this input is non-empty, the current time after processing the input will be greater than $0$, and thus can be expressed as $n+1$.\\[2pt] 
$\mathit{CurPot_N}(i::\mathit{inp}, \mathit{len}) \equiv
\begin{cases}
 \mathit{potential}(w_N, i, \mathit{len}) & \mathit{if}~\tau_N \leq \mathit{CurPot}_N(\mathit{inp}, \mathit{len}), \\
 \mathit{potential}(w_N, i, \mathit{len})\ + & \mathit{otherwise}.\\
\quad \mathit{lk}_N \cdot \mathit{CurPot}_N(\mathit{inp}, \mathit{len}) & \\
\end{cases}$.\\[2pt]
We remark that at a time $n+ 1$, there may be a residual potential from the neuron at time $n$, but only in the case when the potential was not greater than the threshold. 
For $N$, a single-input neuron, and $\mathit{id}$, the identifier of the source of input of $N$, this definition can be simplified and expressed as the lemma below.
\end{rem}

\begin{lem}[Simplified $\mathit{CurPot}_N$ Function for Single-Input Neurons]\ \\
~~\\[-5pt]
$\mathit{CurPot_N}(i::\mathit{inp}, \mathit{len}) \equiv
\begin{cases}
 w_N(\mathit{id}) & \mathit{if}~\tau_N \leq \mathit{CurPot}_N(\mathit{inp}, \mathit{len})\\
 &\quad\mathit{and}~i = \mathit{true}, \\
  0 & \mathit{if}~\tau_N \leq \mathit{CurPot}_N(\mathit{inp}, \mathit{len})\\
 &\quad\mathit{and}~ i =\mathit{false}, \\
 w_N(\mathit{id})~ + & \mathit{if}~ \mathit{CurPot}_N(\mathit{inp}, \mathit{len}) < \tau_N \\
\quad \mathit{lk}_N \cdot \mathit{CurPot}_N(\mathit{inp}, \mathit{len})  &\quad\mathit{and}~ i =\mathit{true},\\ \mathit{lk}_N \cdot \mathit{CurPot}_N(\mathit{inp}, \mathit{len}) & \mathit{otherwise}.\\
\end{cases}$.
\label{lem:curpotoiint}
\end{lem}
If we suppose that $\tau_N \leq w_N(id)$ and \texttt{N} is initial, then there is never a residue of potential. Indeed, at time $0$, there is no residue potential since we start at a potential equal to $0$. At time $n$, if the neuron does not fire and there is no residual potential from before time $n$, that means that the potential at $n$ was null and thus there is no residual potential at time $n+1$.

The first property of Lemma~\ref{lem:sioutcurprops} is a summary of this result.
\begin{lem}[Simplified $\mathit{CurPot}_N$ and $\mathit{Output}_N$ Functions, Weight Reaching Threshold]\ \\
~~\\[-5pt]
\begin{enumerate}
\item $\forall (N:\mathit{Neuron}) (\mathit{id}:\mathit{nat}) (\mathit{inp}:\mathit{list}~\mathit{bool}) (\mathit{len}:\mathit{nat}) (\mathit{i}:\mathit{bool}),\\
  \quad \mathit{One\_input}(N, \mathit{id}, \mathit{len})\land \mathit{is\_initial}_\mathit{Neur}(N,\mathit{len})\land \tau_N\leq w_N(\mathit{id})\rightarrow\\
   \quad\quad \mathit{CurPot}_N(i::\mathit{inp}, \mathit{len})\equiv\begin{cases}
  w_N (\mathit{id}) & \mathit{if}~i = \mathit{true}, \\
  0 & \mathit{otherwise}.\\
  \end{cases}.$
  \\ \texttt{[CurPot\_cons\_w\_greater\_tau\_oi]}
  \item $\forall (N:\mathit{Neuron}) (\mathit{id}:\mathit{nat}) (\mathit{inp}:\mathit{list}~\mathit{bool}) (\mathit{len}:\mathit{nat}) (\mathit{i}:\mathit{bool}),\\
  \quad \mathit{One\_input}(N, \mathit{id}, \mathit{len})\land \mathit{is\_initial}_\mathit{Neur}(N,\mathit{len})\land \tau_N\leq w_N(\mathit{id})\rightarrow\\
   \quad\quad \mathit{Output}_N(i::\mathit{inp}, \mathit{len})= i :: \mathit{Output}_N(\mathit{inp}, \mathit{len})$
  \\ \texttt{[Output\_cons\_w\_greater\_tau\_oi]}
\end{enumerate}
\label{lem:sioutcurprops}
\end{lem}
\noindent
The second property of Lemma~\ref{lem:sioutcurprops} is a consequence of the first one.
By Lemma~\ref{lem:nsteps-props}, we know that the firing of the neuron $N$ at time $n+1$ depends only of the value of $w_N(id)$ and the value of $i$. Since $w_N(id) > \tau_N$, if the value of $i$ is true, the neuron fires. If $i$ is false, since the threshold is strictly positive, the threshold is not reached. The neuron does not fire.

\subsubsection{Delayer Effect}

An important property of a single-input neuron is called the \emph{delayer effect}.
It assumes that an initialized single-input neuron $N$ has a source identified by $id$, and the weight of the neuron $N$ for $id$ is greater than or equal than the neuron’s activation threshold.
In those conditions, the output of $N$ correspond to the inputs provided to $N$ by a delay of one time unit. 
For example, if $N$ receives the input sequence $1001101010$ (written in backwards order), the output produced is $10011010100$.
Neurons with this property primarily act as signal transmitters.
Humans have neurons of this type in their auditory system, associated with chemical synapses.
This property is formalized as Proposition~\ref{lem:delay} below.

\begin{prop}[Delayer Effect for a Single-Input Neuron]\texttt{[Delayer\_Effect]}\\\ \\
$\begin{array}{l}
  \forall (N:\mathit{neuron}) (\mathit{id}:\mathit{nat}) (\mathit{inp}:\mathit{list}~\mathit{bool}) (\mathit{len}:\mathit{nat}),\\
  \quad \mathit{One\_input}(N, \mathit{id}, \mathit{len}) ~\land ~
  \mathit{is\_initial}_{\mathit{Neur}}(N, \mathit{len}) ~\land ~
   w_N(\mathit{id})\geq \tau(N)\rightarrow\\
   \quad\quad\mathit{Output}_N(\mathit{inp}, \mathit{len})=\mathit{inp}~\mathop{++}~[0].
  \end{array}$
\label{lem:delay}
\end{prop}
\begin{proof}
The proof proceeds by induction on the structure of the input sequence.\\
\textbf{Base case}: $\mathit{inp}=[\ ]$ (the empty list).

Since the input sequence is empty and by the properties on $\mathit{is\_initial}_\mathit{Neur}$ in Lemma~\ref{lem:circprops} (3): $\mathit{Output}_N([\ ], \mathit{len}) = \mathit{Output}(N) = [0] = \mathit{inp} ~\mathop{++}~ [0]$.\\
\textbf{Induction case}: We assume that the property holds for $\mathit{inp}$, and we must show that it also holds for $\mathit{i::inp}$, where $i$ is an additional input value.\\
By Lemma \ref{lem:sioutcurprops},
$\mathit{Output}_N(i::\mathit{inp}, \mathit{len}) = 
i:: \mathit{Output}_N (\mathit{inp}, \mathit{len}).$
By the induction hypothesis, $\mathit{Output}_N (\mathit{inp}, \mathit{len}) = \mathit{inp} ~\mathop{++}~ [0]$. With it, the property is verified : $\mathit{Output}_N(i::\mathit{inp}, \mathit{len}) = i:: \mathit{inp} ~\mathop{++}~ [0]$.
\end{proof}

\subsubsection{Filtering Effect}
The second property on a single-input neuron is the \emph{filtering effect}. In contrast to the delayer effect theorem, this property examines what happens when the weights for the external source of input does not reach the threshold of the neuron of interest. In such a case, a single input alone is insufficient to generate a potential that reaches the threshold at that specific time. Consequently, firing may only occur if there are additional positive residual potentials accumulated from previous times. As a result, the neuron cannot fire in two consecutive time steps—this is the essence of this property.
\begin{prop}[Filtering Effect for a Single-Input Neuron]\texttt{[Filtering\_Effect]}\\\ \\
$\begin{array}{l}
  \forall (N:\mathit{Neuron}) (\mathit{id}:\mathit{nat}) (\mathit{inp}:\mathit{list}~\mathit{bool}) (\mathit{len}:\mathit{nat}),\\
  \quad \mathit{One\_input}(N, \mathit{id}, \mathit{len}) ~\land ~
  \mathit{is\_initial}_{\mathit{Neur}}(\mathit{N}, \mathit{len}) ~\land ~
   w_N(\mathit{id}) < \tau(N)\rightarrow\\
   \quad \forall (a_1~a_2 : \mathit{bool}) (l_1~l_2 :\mathit{list}~\mathit{bool}), \\
   \quad\quad \mathit{Output}_N(\mathit{inp}, \mathit{len})=l_1~\mathop{++}~ [a_1; a_2] ~\mathop{++}~ l_2 \rightarrow \\
   \quad\quad\quad a_1 = 0 \lor a_2 = 0.
  \end{array}$
\label{lem:filter}
\end{prop}

\begin{proof}
The proof proceeds by induction on the structure of the input sequence.\\
\textbf{Base case 1}: $\mathit{inp}=[\ ]$ (the empty list).\\
Since the input sequence is empty and by the properties on $\mathit{is\_initial}_\mathit{Neur}$ in Lemma~\ref{lem:circprops} (3): $\mathit{Output}_N([\ ], \mathit{len}) = \mathit{Output}(N) = [0]$. This is in contradiction with one of the hypotheses, thus the property is verified for this case.\\
\textbf{Base case 2}: $\mathit{inp}=[i]$ (one element list).\\
With the same reasoning as previous case and the unfolding property of Lemma~\ref{lem:nsteps-props} (2), we have: $\mathit{Output}_N([i], \mathit{len}) = a\mathop{::} \mathit{Output}_N([\ ],\mathit{len}) = a \mathop{::} [0]$ where $a$ is a boolean value. This verifies the property with $a_1$ as $a$ and $a_2$ as $0$.\\
\textbf{Induction case}: We assume that the property holds for $i_2 :: \mathit{inp}$, and we must show that it also holds for $i_1 :: i_2 ::\mathit{inp}$, where $i_1$ and $i_2$ are input values and $\mathit{inp}$ is an input sequence.\\
By Lemma~\ref{lem:nsteps-props} (2), \\
$\mathit{Output}_N(i_1 :: i_2 ::\mathit{inp}, \mathit{len}) = 
(\tau_N \leq_? \mathit{CurPot}_N(i_1:: i_2 ::\mathit{inp}, \mathit{len})):: \mathit{Output}_N (i_2 :: \mathit{inp}, \mathit{len})$.\\
If $l_1$ contains at least an element, the property is verified with the induction hypothesis. Moving forward, $l_1$ is assumed empty.
By Lemma~\ref{lem:nsteps-props} (2):\\
$\mathit{Output}_N(i_1 :: i_2 ::\mathit{inp}, \mathit{len}) = 
(\tau_N \leq_? \mathit{CurPot}_N(i_1:: i_2 ::\mathit{inp}, \mathit{len}))::(\tau_N \leq_? \mathit{CurPot}_N(i_2 ::\mathit{inp}, \mathit{len})):: \mathit{Output}_N (\mathit{inp}, \mathit{len})$.\\
$a_1$ and $a_2$ are in this case respectively : $(\tau_N \leq_? \mathit{CurPot}_N(i_1:: i_2 ::\mathit{inp}, \mathit{len}))$ and $(\tau_N \leq_? \mathit{CurPot}_N(i_2 ::\mathit{inp}, \mathit{len}))$, and $l_2$ is $\mathit{Output}_N (\mathit{inp}, \mathit{len}) $.

The proof is split in the two sub-cases on the value of $\tau_N \leq_? \mathit{CurPot}_N(i_2 ::\mathit{inp}, \mathit{len})$. If the value is false, the property is verified.

\textbf{Sub-case $\tau_N \leq_? \mathit{CurPot}_N(i_2 ::\mathit{inp},\mathit{len}) = \mathit{true}$}.
By Lemma~\ref{lem:nsteps-props} (1), 
\[
\mathit{CurPot}_N(i_1::i_2::\mathit{inp}, \mathit{len}) \equiv 
\begin{cases}
  potential(w_N, i_1, \mathit{len}) & \text{if } \tau_N \leq \mathit{CurPot}_N(i_2::\mathit{inp}, \mathit{len}), \\
  potential(w_N, i_1,\mathit{len})\ + & \text{otherwise}.\\
\quad lk_N \cdot \mathit{CurPot}_N(i_2::\mathit{inp}, \mathit{len}) & \\
\end{cases}
\]
With the hypothesis of this subcase, we have:\\
$\mathit{CurPot}_N(i_1::i_2::\mathit{inp}, \mathit{len}) \equiv potential(w_N, i_1, \mathit{len})$.\\
By Lemma~\ref{lem:sineurprops} (2), using $\tau_N$ as $m$, we have $\tau_N > \mathit{potential}(w_N, i_1, \mathit{len})$.
Thus $(\tau_N \leq_?\mathit{CurPot}_N(i_1::i_2::\mathit{inp}, \mathit{len})) = (\tau_N \leq_? \mathit{potential}(w_N, i_1, \mathit{len})) = \mathit{false}$. The property is verified.
\end{proof}

\subsubsection{Single-Input General Behaviors}
With Propositions~\ref{lem:delay} and~\ref{lem:filter}, we gain a more generalized understanding of the behavior of an initialized single-input neuron. This property encompasses both previous properties: an initialized single-input neuron either fires with a one-unit time delay or is unable to fire in two consecutive time steps.
\begin{cor}[General Behavior for a Single-Input Neuron]\texttt{[One\_input\_generic]}\\\ \\
$\begin{array}{l}
  \forall (N:\mathit{Neuron}) (\mathit{id}:\mathit{nat}) (\mathit{inp}:\mathit{list}~\mathit{bool}) (\mathit{len}:\mathit{nat}),\\
  \quad \mathit{One\_input}(N, \mathit{id}, \mathit{len}) ~\land ~
  \mathit{is\_initial}_{\mathit{Neur}}(N, \mathit{len}) \rightarrow \\
  \quad\quad\mathit{Output}_N(\mathit{inp}, \mathit{len})=\mathit{inp}~\mathop{++}~[0] ~ \lor\\
   \quad \quad \forall (a_1~a_2 : \mathit{bool}) (l_1~l_2 :\mathit{list}~\mathit{bool}), \\
   \quad\quad \quad \mathit{Output}_N(\mathit{inp}, \mathit{len})=l_1~\mathop{++}~ [a_1; a_2] ~\mathop{++}~ l_2 \rightarrow \\
   \quad\quad\quad\quad a_1 = 0 \lor a_2 = 0.
  \end{array}$
\label{lem:genebehav}
\end{cor}

\subsubsection{Property: Spike Decreasing}

The final property for single-input neurons concerns the total number of spikes produced by an initialized single-input neuron, that is, the number of times the neuron fires.
When a single-input neuron receives a $0$ input, its potential does not increase, preventing it from firing.
This property, called the \emph{spike decreasing property}, states that the number of firings is at most equal to the number of $1$s in the input sequence.
We denote the number of occurrences of $1$ in a list $l$ by the notation $\mathit{number\_occ}(l, 1)$.\footnote{This function is \texttt{count\_occ} in Rocq's list library.}

\begin{prop}[Spike Decreasing Behavior for a Single-Input Neuron]\texttt{[Spike\_Decreasing]}\\\ \\
$\begin{array}{l}
  \forall (N:\mathit{Neuron}) (\mathit{id}:\mathit{nat}) (\mathit{inp}:\mathit{list}~\mathit{bool}) (\mathit{len}:\mathit{nat}),\\
  \quad \mathit{One\_input}(N, \mathit{id}, \mathit{len}) ~\land ~
  \mathit{is\_initial}_{\mathit{Neur}}(N, \mathit{len}) \rightarrow \\
  \quad\quad\mathit{number\_occ}(\mathit{Output}_N(\mathit{inp}, \mathit{len}), 1) \leq \mathit{number\_occ}(\mathit{inp}, 1).
  \end{array}$
\label{lem:spike}
\end{prop}

\begin{proof}
The proof proceeds by induction on the structure of the input sequence.\\
\textbf{Base case}: $\mathit{inp}=[\ ]$ (the empty list). \\ Since the input sequence is empty and by the properties on $\mathit{is\_initial}_\mathit{Neur}$ in Lemma~\ref{lem:circprops} (3): $\mathit{number\_occ}(\mathit{Output}_N([\ ], \mathit{len}),1) = \mathit{number\_occ}(\mathit{Output}(N),1) = \mathit{number\_occ}([0],1) = 0 \leq 0 = \mathit{number\_occ}([\ ],1)$. The property is verified for this case.\\
\textbf{Induction case}: We assume that the property holds for $\mathit{inp}$, and we must show that it also holds for $\mathit{i::inp}$, where $i$ is an additional input value.\\
By Lemma \ref{lem:nsteps-props} (2),
$\mathit{Output}_N(i::\mathit{inp}, \mathit{len}) = 
(\tau_N \leq_? \mathit{CurPot}_N(i ::\mathit{inp}, \mathit{len})):: \mathit{Output}_N (\mathit{inp}, \mathit{len}).$
By the induction hypothesis, 
$\mathit{number\_occ}(\mathit{Output}_N(\mathit{inp}, \mathit{len}),1) \leq \mathit{number\_occ}(\mathit{inp},1)$.\\
If we prove that $\mathit{number\_occ}((\tau_N \leq_? \mathit{CurPot}_N(i ::\mathit{inp}, \mathit{len})),1) \leq \mathit{number\_occ}(i,1)$, the property is verified.\\
\textbf{Sub-case 1 : $i = 1$}.\\
$\mathit{number\_occ}(i,1) =1$ so the inequality is verified.\\
\textbf{Sub-case 2 : $i = 0$}.\\
$\mathit{number\_occ}(i,1) =0$.
By Lemma~\ref{lem:curpotoiint} with $i = 0$, we have :\\ $\mathit{CurPot_N}(i::\mathit{inp}, \mathit{len}) \equiv
\begin{cases}
   0 & \mathit{if}~\tau_N \leq \mathit{CurPot}_N(\mathit{inp}, \mathit{len})\\
\mathit{lk}_N \cdot \mathit{CurPot}_N(\mathit{inp}, \mathit{len}) & \mathit{otherwise}.\\
\end{cases}$.\\
If $\tau_N \leq \mathit{CurPot}_N(\mathit{inp},\mathit{len})$, the property is verified.\\
\textbf{Sub-sub-case : $\tau_N > \mathit{CurPot}_N(\mathit{inp},\mathit{len})$}.\\
Since $0 \leq lk_N \leq 1$, we know that $\tau_N > lk_N \cdot \mathit{CurPot}_N(\mathit{inp},\mathit{len})$.
We thus have:\\
$\mathit{number\_occ}((\tau_N \leq_? \mathit{CurPot}_N(i ::\mathit{inp}, \mathit{len})),1)=\\
\mathit{number\_occ}(\tau_N \leq_? \mathit{lk}_N \cdot \mathit{CurPot}_N(\mathit{inp},\mathit{len}),1) = 0$.
The property is verified.
\end{proof}

\section{Properties of Circuits and their Proofs}
\label{sec:circuitprops}

All the properties from Section~\ref{sec:neuronprops} are translatable to the circuit.  
In this section, we do not prove properties specific to circuits, but instead prove an equivalence that allows us to conclude that the properties from the previous section on single neurons hold for every neuron in a circuit.
The main functions for which this equivalence is defined are the \texttt{AfterNstepsNeuron} function in Figure~\ref{fig:updateneuron} and the \texttt{curpot\_neuron} and \texttt{output\_neuron} functions from Figure~\ref{fig:circuitfuns}; the latter two call the \texttt{Nstepscircuit} function from Figure~\ref{fig:circuitinputs} to do the main work.
The properties expressing this equivalence appear in Figure~\ref{fig:Coqcuroutequiv}.
\begin{figure}[htb]
\begin{verbatim}
Lemma curpot_neuron_nstep : 
    forall (nc : Neurocircuit) (inp : list (nat -> bool)) (n : Neuron),
        In n (ListNeuro nc) ->
        curpot_neuron nc inp (Id (Feature n)) == 
        CurPot (AfterNstepsNeuron n (inp_mult inp nc) 
                (length (ListNeuro nc) + SupplInput nc)).
    
Lemma output_neuron_nstep : forall nc inp n,
    forall (nc : Neurocircuit) (inp : list (nat -> bool)) (n : Neuron),
        In n (ListNeuro nc) ->
        output_neuron nc inp (Id (Feature n)) = 
        Output (AfterNstepsNeuron n (inp_mult inp nc) 
                (length (ListNeuro nc) + SupplInput nc)).
\end{verbatim}
\caption{Equivalence between current potential and output definition lemmas in Rocq}
\label{fig:Coqcuroutequiv}
\end{figure}
We state these properties directly in Rocq instead of using our mathematical notation, mainly because they are direct statements about Rocq functions given in the figures just mentioned, but also because they provide an example illustration of the properties expressed directly in Rocq.
The essence of the equivalence is that processing of inputs done by \texttt{AfterNstepsNeuron} is the same as that done by \texttt{Nstepscircuit} for all the neurons in a circuit, modulo some pre-processing of the input.
In particular, the \texttt{inp\_mult} function appearing in the lemmas in the figure (whose definition is omitted)
adjusts the input sequence, originally defined for a neuron within a circuit, to an equivalent input sequence in a context where circuit boundaries are not specified. Indeed, in the case of a circuit, the input sequence must provide the values for inputs originating from external sources for the functions \texttt{curpot\_neuron} and \texttt{output\_neuron}. Internal input sources correspond to the outputs of other neurons within the circuit at the previous unit of time and are thus already precomputed at each time step.
In contrast to the case of an individual neuron treated as a distinct unit, the input sequence for \texttt{AfterNstepsNeuron} includes inputs from every neuron in the environment. Here, there is no distinction between internal and external inputs, as no circuit boundaries are defined. The function call \texttt{(inp\_mult\ inp\ nc)} incorporates inter-neuron inputs within the circuit by using each neuron's output from the previous time step.

We give a few examples of corollaries that follow from the lemmas in Figure~\ref{fig:Coqcuroutequiv}.
First, consider Lemma~\ref{lem:nonnegcurr} stating that for a multiple-input neuron $N$ in an environment of size $\mathit{len}$, if all values of inputs $\mathit{inp}$ are non-negative, then the value of the $\mathit{CurPot}_N(\mathit{inp},\mathit{len})$ is also non-negative.
Recall that $\mathit{CurPot}_N(\mathit{inp},\mathit{len})$ is notation for \texttt{(CurPot (AfterNstepsNeuron N inp len))}.
Using the lemmas in Figure~\ref{fig:Coqcuroutequiv}, the circuit version of this lemma, stated below, follows as a corollary.
\begin{cor}[$\mathit{curpot}_\mathit{NC}$ Always Non-Negative for Multiple-Input Neuron] \texttt{[AlwaysNNegNC]}\\\ \\
$\begin{array}{l}
  \forall (N:\mathit{Neuron})  (\mathit{inp}:\mathit{list} ~ (\mathit{nat}\rightarrow\mathit{bool})) (\mathit{NC}:\mathit{NeuroCircuit}),\\
  \quad\mathit{is\_initial_\mathit{Cir}}(\mathit{NC}) \land
  N\in\mathit{ln}_\mathit{NC} \land{}\\
  \quad(\forall (\mathit{id}:\mathit{nat}), \mathit{id} < \mathit{length}(\mathit{ln}_\mathit{NC})+\mathit{si}_\mathit{NC} \rightarrow w_N(\mathit{id})\geq 0)\rightarrow{}\\
  \quad\quad \mathit{curpot}_\mathit{NC}(N,\mathit{inp})\geq0.
  \end{array}$
\label{cor:nonnegcurr}
\end{cor}
\noindent
Recall that $\mathit{curpot}_\mathit{NC}(N,\mathit{inp})$ is notation for \texttt{(curpot\_neuron NC inp (Id (Feature N)))}.

In the case of properties about single-input neurons, the corollaries distinguish between scenarios where the single-input source is another neuron within the circuit or an external source.   
For example, the delayer effect for a single-input neuron is expressed in Proposition~\ref{lem:delay} and the following corollary expresses this property for the two cases for circuits.
\begin{cor}[Delayer Effect for Internal and External Input Sources]\ \\
~~\\[-5pt]
\begin{enumerate}
\item $\forall (N_1~N_2:\mathit{Neuron}) (\mathit{inp}:\mathit{list} ~ (\mathit{nat}\rightarrow\mathit{bool})) (\mathit{NC}:\mathit{NeuroCircuit}),\\
  \qquad\mathit{is\_initial_\mathit{Cir}}(\mathit{NC}) \land
  N_1,N_2\in\mathit{ln}_\mathit{NC} \land{}\\
  \mathit{One\_input}(N_1,\mathit{id}_{N_2},\mathit{length}(\mathit{ln}_\mathit{NC})+\mathit{si}_\mathit{NC})\land
  \tau_{N_1}\le w_{N_1}(\mathit{id}_{N_2})\rightarrow{}\\
  \mathit{output}_\mathit{NC}(N_1,\mathit{inp})=
  \mathit{tl}(\mathit{output}_\mathit{NC}(N_2,\mathit{inp}))~\mathop{++}~[0]$
  \\ \texttt{[One\_Input\_NC\_delay\_Int]}
\item $\forall (N:\mathit{Neuron})~(m:\mathit{nat})~(\mathit{inp}:\mathit{list} ~ (\mathit{nat}\rightarrow\mathit{bool})) (\mathit{NC}:\mathit{NeuroCircuit}),\\
  \mathit{is\_initial_\mathit{Cir}}(\mathit{NC}) \land
  N\in\mathit{ln}_\mathit{NC} \land
  \mathit{length}(\mathit{ln}_\mathit{NC}) \le m < \mathit{length}(\mathit{ln}_\mathit{NC})+\mathit{si}_\mathit{NC} \land{}\\
  \mathit{One\_input}(N,m,\mathit{length}(\mathit{ln}_\mathit{NC})+\mathit{si}_\mathit{NC})\land
  \tau_N\le w_N(m)\rightarrow{}\\
  \mathit{output}_\mathit{NC}(N,\mathit{inp})=
  \mathit{map}~(\mathit{fun}~f\Rightarrow f\ m)~\mathit{inp} ~\mathop{++}~ [0]
  $\\ \texttt{[One\_Input\_NC\_delay\_Ext]}
\end{enumerate}
\label{cor:delay}
\end{cor}
\noindent
Corollary~\ref{cor:delay} (1) considers the case where the output of a neuron inside the circuit serves as an input to another neuron in the circuit.  
In particular, the output of $N_2$ serves as input to $N_1$, which means $N_1$ has a delay of one time unit before producing the output generated by $N_2$; this is expressed by taking the tail of the output of $N_2$ and adding an additional $0$ at the end. 
Part (2) expresses the case where an external input to the circuit is connected directly to neuron $N$.
In this case, the output of $N$ is the same as the input with an additional $0$ at time $0$.

\section{Properties of Archetypes and their Proofs}
\label{sec:archetypeprops}
In each of the following subsections, we present representative properties for the majority of archetypes shown in Figure~\ref{fig:archetypes}, focusing on how the choice of weights and inputs affects the output. 
For \texttt{Series} and \texttt{ParallelComposition}, we examine the delayer effect on the neurons within these circuits. 
For \texttt{PositiveLoop}, we analyze the impact on the output of different input sequences coming from external sources.
Lastly, for \texttt{NegativeLoop} and \texttt{ContraInhib}, we focus on generating specific patterns in the output sequence by using external input sources set to $\mathit{true}$ and varying the weights of individual neurons.
The delayer effect (Corollary~\ref{cor:delay}) plays a significant role in the following properties since many neurons occurring in archetypes are single-input neurons.

\subsection{Simple Series with and without Multiple Outputs}
\label{subsec:serie}
Recall that the archetypes of simple series and series with multiple outputs are both represented by the record \texttt{Series}. 
As we have seen previously, if the weight of a single-input neuron reaches the threshold for the source of input, the input sequence is delayed by one time unit and given as output.  
In this section, we study the delayer effect of a series and consider the output of every neuron in these circuits.
Recall that by definition of \texttt{Series}, in a series of length $n$, the identifiers of the neurons will be in the range $0,\ldots,n-1$ and the output of each neuron (except the last) with identifier $\mathit{id}$ will serve as input to neuron $\mathit{id}+1$.
Specifically, in the delayer effect for a series, each neuron having identifier $\mathit{id}$ delays the external input sequence by $\mathit{id}+1$ time units if the non-null weight of each neuron reaches the threshold.
Here, if the series has three neurons, for example, the neuron with identifier $2$ has $3$ units of delay and thus 
the full input sequence will not appear as output; it will be truncated.
If we consider the example input sequence $011010111$ (in backward order as usual), the output of the third neuron is $1010111000$, and with the shorter input sequence $01$, the output is $000$.
As expressed in Proposition~\ref{lem:seriedelay} below, the truncation process either completely removes the input sequence $\mathit{inp}$—if the neuron's identifier exceeds the length of the sequence—resulting in an output sequence of $\mathit{length}(\mathit{inp}) + 1$ zeros, or partially removes a number of elements from the beginning of the input sequence equal to the neuron's identifier.

We simplify the type of the external input sequence in this subsection (and in~\ref{subsec:paracompo},~\ref{subsec:posloop}, and~\ref{subsec:negloop}), again using $\mathit{list}~\mathit{bool}$ instead of $\mathit{list}~(\mathit{nat}\rightarrow\mathit{bool})$, since the properties in these sections are about archetypes that have only one external source of input.
Going forward, we denote the neuron in an archetype having the identifier $\mathit{id}$ as $N_\mathit{id}$, 
and if a neuron $N_\mathit{id}$ has an external input source, we use the notation $ext_\mathit{id}$ to represent this external neuron.
For example, in a series of length $n$, the identifier of the external source is $n$ and it is connected to $N_0$, so $\mathit{ext}_0$ represents $N_n$.
The notation $\mathit{repeat}(v, n)$ denotes a list of $n$ elements where each element is $v$.
The notation $l[i:~]$, where $l$ is a list and $i$ is a natural number, denotes the sublist of $l$ starting at position $i$, assuming the first position is $0$.\footnote{The \texttt{repeat} function is in Rocq's list library. The $[~:~]$ function is also in the list library, called \texttt{afternlist}.}
\begin{prop}[Delayer Effect in Series] \texttt{[Series\_Delayer\_Effect]}\\\ \\
$\begin{array}{l}
  \forall (\mathit{NC}: \mathit{NeuroCircuit}) (N:\mathit{Neuron}) (\mathit{inp}:\mathit{list}~\mathit{bool}),\\
  \quad \mathit{Series}(\mathit{NC}) ~ \land ~ N \in \mathit{ln}_\mathit{NC} ~\land ~
  \mathit{is\_initial}_\mathit{Cir}(\mathit{NC}) ~\land ~ w_{N_{0}}(\mathit{ext}_0) \geq \tau_{N_{0}}~\land ~ \\
  \quad(\forall (i: \mathit{nat}), i + 1 < \mathit{length}({\mathit{ln}_\mathit{NC}}) \rightarrow w_{N_{i+1}}(i) \geq \tau_{N_{i+1}})\rightarrow \\
  \quad\mathit{output}_\mathit{NC}(N, \mathit{inp}) = 
  \begin{cases}
  \mathit{inp}[\mathit{id}_N:~]~\mathop{++}~\mathit{repeat}(0, \mathit{id}_N+1) & \mathit{if}~\mathit{id}_N \leq \mathit{length} (\mathit{inp}) \\
 \mathit{repeat}(0,\mathit{length}(\mathit{inp})+1)& \mathit{otherwise}.
  \end{cases}
  \end{array}$
\label{lem:seriedelay}
\end{prop}
\begin{proof}
Note that the delayer effect property (Corollary~\ref{cor:delay})
applies to each neuron in the series $\mathit{NC}$, as all inter-neuron weights and those from the external source exceed their respective thresholds and each neuron are initial single-input neurons.
We also remark that the proof is for a specific external input sequence $\mathit{inp}$ and thus for a specific time
$t_\mathit{NC}$, which is the time step after all input has been processed.
Recall that constraint \texttt{TimeNeuro} on circuits ensures that $t_\mathit{NC}$ is one less that the length of the output of all of the neurons in the circuit.
We proceed by induction on the identifier number of each neuron.

\noindent
\textbf{Base case: Neuron $N_0$ with identifier $0$}.
We apply the delayer effect property (Corollary~\ref{cor:delay}) to the first neuron of the series: $output_{\mathit{NC}}(N_0, inp) = inp ~\mathop{++}~ [0]$. Since $0 \leq length(inp)$, the property is verified.

\noindent
\textbf{Induction case}: We assume that the property holds for the neuron $N_{id}$ whose identifier is $\mathit{id}$, and we must show that it also holds for the neuron $N_{id+1}$ whose identifier is $\mathit{id} +1$.

If the time $t_\mathit{NC}$ is equal to $0$ ($\mathit{inp} =[\ ]$), the output of $N_{id + 1}$ is $[0]$ since $N_{id + 1}$ is an initial neuron and thus the property is verified.

We suppose now that $t_\mathit{NC} \neq 0$.
The input sequence for neuron $N_{id+1}$ at time $t_\mathit{NC}$ is $output_{\mathit{NC}}(N_\mathit{id}, inp)[1:\ ]$. The last element of the output list of neuron $N_\mathit{id}$ is not processed by neuron $N_{\mathit{id}+1}$. Indeed, this element is produced at time $t_\mathit{NC}$ by $N_\mathit{id}$ and cannot be provided at the same time to $N_{id+1}$.
We apply the delayer effect property (Corollary~\ref{cor:delay}) to neuron $N_{id+1}$:

$\begin{array}{l}
output_{\mathit{NC}}(N_{id+1}, inp) = output_{\mathit{NC}}(N_{id}, inp)[1:\ ] ~\mathop{++}~ [0].
\end{array}$

\noindent
By the induction hypothesis:

$\mathit{output}_\mathit{NC}(N_{id}, \mathit{inp}) = 
  \begin{cases}
  \mathit{inp}[\mathit{id}:~]~\mathop{++}~\mathit{repeat}(0, \mathit{id}+1) & \mathit{if}~\mathit{id} \leq \mathit{length} (\mathit{inp}) \\
 \mathit{repeat}(0,\mathit{length}(\mathit{inp})+1)& \mathit{otherwise}.
  \end{cases}$

\noindent
If $\mathit{id} + 1\leq \mathit{length} (\mathit{inp})$, then we have:

$\begin{array}{l}
(\mathit{inp}[\mathit{id}:~]~\mathop{++}~\mathit{repeat}(0, \mathit{id}+1))[1:\ ] ~\mathop{++}~ [0] = {}\\
\quad\mathit{inp}[\mathit{id} +1:~]~\mathop{++}~\mathit{repeat}(0, \mathit{id}+1) ~\mathop{++}~ [0] = {}\\
\quad\quad\mathit{inp}[\mathit{id} +1:~]~\mathop{++}~\mathit{repeat}(0, \mathit{id}+2).
\end{array}$

\noindent
If $\mathit{id} = \mathit{length} (\mathit{inp})$, then we have:

$\begin{array}{l}(\mathit{inp}[\mathit{id}:~]~\mathop{++}~\mathit{repeat}(0, \mathit{id}+1))[1:\ ] ~\mathop{++}~ [0] = {}\\
\quad(\mathit{inp}[\mathit{length}(\mathit{inp}):~]~\mathop{++}~\mathit{repeat}(0, \mathit{length}(\mathit{inp}) +1))[1:\ ] ~\mathop{++}~ [0] = {}\\
\quad\quad[\ ]~\mathop{++}~\mathit{repeat}(0, \mathit{length}(\mathit{inp})+1)[1:\ ]~\mathop{++}~[0] = {}\\
\quad\quad\quad\mathit{repeat}(0, \mathit{length}(\mathit{inp})+1).
\end{array}$

\noindent
Thus, we have:

$\mathit{output}_\mathit{NC}(N_{id+1}, \mathit{inp}) = 
  \begin{cases}
  \mathit{inp}[\mathit{id+1}:~]~\mathop{++}~\mathit{repeat}(0, \mathit{id}+2) & \mathit{if}~\mathit{id} + 1 \leq \mathit{length} (\mathit{inp}) \\
 \mathit{repeat}(0,\mathit{length}(\mathit{inp})+1)& \mathit{otherwise}.
  \end{cases}$

\noindent  
Note that the case when $\mathit{id} > \mathit{length} (\mathit{inp})$ falls into the ``otherwise'' case.  Thus, the property is verified.
\end{proof}

\subsection{Parallel Composition} 
\label{subsec:paracompo}
As mentioned, the property we consider about parallel composition also concerns the delayer effect. Unlike in the series, the delayer effect in parallel composition impacts the circuit's outputs as follows: the first neuron delays the input sequence by one time unit, while the other neurons introduce a delay of two time units.
\begin{prop}[Delayer Effect in Parallel Composition] 
~~\\[-5pt]
\begin{enumerate}
\item $\forall (\mathit{NC}: \mathit{NeuroCircuit}) (N:\mathit{Neuron}) (\mathit{inp}:\mathit{list}~\mathit{bool}),\\
  \mathit{ParallelComposition}(\mathit{NC}) ~\land ~ N \in \mathit{ln}_\mathit{NC} ~\land ~
  \mathit{is\_initial}_\mathit{Cir}\mathit{(NC)} ~\land ~ \mathit{id}_N = 0 ~\land\\ ~ w_{N_{0}}(\mathit{ext}_0) \geq \tau_{N_{0}}~\rightarrow~
  \mathit{output}_\mathit{NC}\mathit{(N, inp)} = inp ~\mathop{++}~ [0].$\\
  \texttt{[ParalComp\_Delayer\_Effect\_0]}
\item $\forall (\mathit{NC}: \mathit{NeuroCircuit}) (N:\mathit{Neuron}) (\mathit{inp}:\mathit{list}~\mathit{bool}),\\
  \quad \mathit{ParallelComposition}(\mathit{NC}) ~ \land ~ N \in \mathit{ln}_\mathit{NC} ~\land ~
  \mathit{is\_initial}_\mathit{Cir}\mathit{(NC)} ~\land ~ ~ \mathit{id}_N \neq 0 ~ \land\\
  \quad w_{N_{0}}(\mathit{ext}_0) \geq \tau_{N_{0}}~\land 
  w_{N}(0) \geq \tau_{N}\rightarrow\\
   \mathit{output}_\mathit{NC}\mathit{(N, inp)} =
  \begin{cases}
  inp[1:~]~\mathop{++}~[0;0] & \mathit{if}\ 0 < len (inp)\\
  [0]& \mathit{otherwise}\\
  \end{cases}$\\
  \texttt{[ParalComp\_Delayer\_Effect\_Succ]}
  \end{enumerate}
\label{prop:paralcomp}
\end{prop}
\begin{proof}
The reasoning for the archetype \texttt{ParallelComposition} is similar to that for the archetype \texttt{Series}. The proof is by induction on the identifier numbers of neurons. The difference occurs for neurons other than the first one. We take such a neuron $N_{id}$ with identifier $id$ distinct from $0$.
If time $t_\mathit{NC}$ is equal to $0$ ($\mathit{inp} =[\ ]$), the output of $N_{id}$ is $[0]$ since $N_{id}$ is an initial neuron, and thus the property is verified.

If time $t_\mathit{NC} \neq 0$, the input sequence for neuron $N_{id}$ at time $t_\mathit{NC}$ is the output list of the first neuron without the last element which corresponds to $\mathit{inp}[1:\ ]~\mathop{++}~[0]$.
When we apply the delayer effect theorem (Corollary~\ref{cor:delay}), we obtain a delay of 2 units of time: $\mathit{inp}[1:\ ]~\mathop{++}~[0;0]$.
\end{proof}

\subsection{Positive Loop}
\label{subsec:posloop}
A positive loop introduces positive feedback into the circuit. Throughout this subsection, we assume that all non-zero neuron weights exceed their respective thresholds. We analyze various types of input sequences and their effects on the outputs of each neuron within the circuit.

The first property focuses on the output behavior of neurons in a positive loop when the external inputs are consistently $0$s. In this scenario, the circuit cannot generate positive feedback and consequently, none of the neurons fire at any time. Indeed, there are no $1$-inputs that increase the potential of any neuron enough to reach its threshold.
Recall that the identifier of the external source of input in this archetype is $2$. 

\begin{prop}[Output Behaviors with an Input Sequence of 0s in a Positive Loop]\ \\ 
\texttt{[PL\_Amplifier\_input\_false]}\\\ \\
$\begin{array}{l}
  \forall (\mathit{NC}: \mathit{NeuroCircuit}) (\mathit{inp}:\mathit{list}~\mathit{bool}),\\
  \quad \mathit{PositiveLoop}(\mathit{NC}) ~ \land  ~
  \mathit{is\_initial}_\mathit{Cir}\mathit{(NC)} ~\land ~ (\forall (b: \mathit{bool}), b \in \mathit{inp} \rightarrow b = 0)~\land \\
  \quad w_{N_{0}}(2) \geq \tau_{N_{0}}~\land ~ w_{N_{0}}(1) \geq \tau_{N_{0}}~\land ~ w_{N_{1}}(0) \geq \tau_{N_{1}}~\rightarrow \\
  \quad\quad \mathit{output}_\mathit{NC}(N_{0}, \mathit{inp}) = \mathit{repeat}(0, \mathit{length} (\mathit{inp})+1) ~ \land ~\\
  \quad\quad\mathit{output}_\mathit{NC}(N_{1}, \mathit{inp}) = \mathit{repeat}(0, \mathit{length} (\mathit{inp})+1).
  \end{array}$
\label{lem:posloop1}
\end{prop}
\begin{proof}
The proof is by induction on the length of the input. If $t_\mathit{NC}$ is 0 ($\mathit{inp} = [\ ]$), the current potential of both neurons is null. Otherwise, when the external input sequence is only zeros, the potential of each neuron always stays the same and is equal to $0$ at every time unit. Thus, the threshold of each neuron is never reached, and thus at each time unit, the output is $0$.
\end{proof}

The second property, which we call the \emph{amplifier property}, focuses on input sequences that begin with zeros, followed by two consecutive ones, and may include additional inputs afterward. The pair of consecutive ones activates the positive loop, causing both neurons to enter a state of perpetual firing, regardless of any subsequent external inputs.

\begin{prop}[Output Behaviors with Input Containing Two Successive $1$s in a Positive Loop]\texttt{[PL\_Amplifier\_input\_2\_true]}\\\ \\
$\begin{array}{l}
  \forall (\mathit{NC}: \mathit{NeuroCircuit}) (\mathit{inp}_1 \mathit{inp}_2:\mathit{list}~\mathit{bool}),\\
  \quad \mathit{PositiveLoop}(\mathit{NC}) ~\land ~
  \mathit{is\_initial}_\mathit{Cir}(\mathit{C}) ~\land ~ (\forall (b: bool), b \in \mathit{inp}_2 \rightarrow b = 0)~\land \\
  \quad w_{N_{0}}(2) \geq \tau_{N_{0}}~\land ~ w_{N_{0}}(1) \geq \tau_{N_{0}}~\land ~ w_{N_{1}}(0) \geq \tau_{N_{1}}~\rightarrow \\
  \quad\quad \mathit{output}_\mathit{NC}(N_{0}, \mathit{inp}_1\mathop{++}[1;1]\mathop{++}\mathit{inp}_2) = \\
  \quad\quad\quad\quad \mathit{repeat}(1, \mathit{length} (\mathit{inp}_1)+2)~\mathop{++}~\mathit{repeat}(0, \mathit{length} (\mathit{inp}_2)+1) ~ \land ~\\
  \quad\quad\mathit{output}_\mathit{NC}(N_{1}, \mathit{inp}_1\mathop{++}[1;1]\mathop{++}\mathit{inp}_2) = \\
  \quad\quad\quad\quad \mathit{repeat}(1, \mathit{length} (\mathit{inp}_1)+1)~\mathop{++}~repeat(0, \mathit{length} (\mathit{inp}_2)+2).
  \end{array}$
\label{prop:posloop2}
\end{prop}
\begin{proof}
The neuron with identifier $1$ is an initialized single-input neuron, with the weight of its sole input source meeting the threshold $\tau_{N_1}$. These properties fulfill the requirements to apply the filtering effect theorem.
The proof proceeds by induction on the structure of the input sequence $\mathit{inp}_1$.

\noindent
\textbf{Base case}: $\mathit{inp}_1 = [\ ]$.

\noindent
\textbf{Neuron $N_0$}: By Proposition~\ref{lem:posloop1}, $\mathit{output}_\mathit{NC}(N_{0}, \mathit{inp}_2) = repeat(0, \mathit{length} (\mathit{inp}_2)+1)$.\\ 
We remark that if the external input at a time unit is $1$, the current potential of $N_0$ reaches its threshold and $N_0$ fires since $w_{N_{0}}(2) \geq \tau_{N_{0}}$ and $w_{N_{0}}(1) \geq 0$.
Thus, we have:

$\begin{array}{l}
\mathit{output}_\mathit{NC}(N_{0}, 
[1;1]~ \mathop{++}~\mathit{inp}_2) = 
 [1;1] ~\mathop{++}~ repeat(0, \mathit{length} (\mathit{inp}_2)+1).
\end{array}$

\noindent
\textbf{Neuron $N_1$}: By the delayer effect theorem (Corollary~\ref{cor:delay}):

$\begin{array}{l}
\mathit{output}_\mathit{NC}(N_{1}, [1;1] ~\mathop{++}~\mathit{inp}_2) = [1] ~\mathop{++}~ repeat(0, \mathit{length} (\mathit{inp}_2)+2).
\end{array}$

\noindent
\textbf{Induction case}: We assume that the property holds for $\mathit{inp_1}$, and we must show that it also holds for $\mathit{i::inp_1}$, where $i$ is an additional input value.
By the induction hypothesis:

$\begin{array}{l}
\mathit{output}_\mathit{NC}(N_{0}, \mathit{inp}_1 ~\mathop{++}~[1;1]~\mathop{++}~\mathit{inp}_2) = {}\\
\quad\mathit{repeat}(1, \mathit{length} (\mathit{inp}_1)+2)~\mathop{++}~\mathit{repeat}(0, \mathit{length} (\mathit{inp}_2)+1) \mbox{ and } \\
\mathit{output}_\mathit{NC}(N_{1}, \mathit{inp}_1 ~\mathop{++}~[1;1]~\mathop{++}~\mathit{inp}_2) = {} \\
\quad\mathit{repeat}(1, \mathit{length} (\mathit{inp}_1)+1)~\mathop{++}~\mathit{repeat}(0, \mathit{length} (\mathit{inp}_2)+2).
\end{array}$

\noindent
\textbf{Neuron $N_0$}: Since $w_{N_{0}}(2) \geq 0$, $w_{N_{0}}(1) \geq \tau_{N_{0}}$, and the last output of $N_1$ is $1$ after processing the input sequence $\mathit{inp}_1~\mathop{++}~[1;1]~\mathop{++}~\mathit{inp}_2$, we can conclude that the current potential of $N_0$ reaches its threshold and the neuron fires. We have: 

$\begin{array}{l}
\mathit{output}_\mathit{NC}(N_{0}, i::\mathit{inp}_1 ~\mathop{++}~[1;1]~\mathop{++}~\mathit{inp}_2) = {}\\
\quad\mathit{repeat}(1, \mathit{length} (\mathit{inp}_1)+3)~\mathop{++}~\mathit{repeat}(0, \mathit{length} (\mathit{inp}_2)+1).
\end{array}$

\noindent
\textbf{Neuron $N_1$}: By the delayer effect theorem (Corollary~\ref{cor:delay}), we have:

$\begin{array}{l}
\mathit{output}_\mathit{NC}(N_{1}, i::\mathit{inp}_1 ~\mathop{++}~[1;1]~\mathop{++}~\mathit{inp}_2) = {}\\
\quad\mathit{repeat}(1, \mathit{length} (\mathit{inp}_1)+2)~\mathop{++}~\mathit{repeat}(0, \mathit{length} (\mathit{inp}_2)+2).
\end{array}$
\end{proof}

The third property is an \emph{oscillation property} that occurs with an input sequence consisting entirely of zeros, except for a single $1$.
In this case, when the input $1$ is processed, the first neuron fires while the second does not.
In the following time unit, the second neuron fires while the first does not.
This firing alternates between the two neurons with each subsequent time unit.

We define $\mathit{repeat\_pattern}(l, n)$ as a function that creates a list of $n$ elements by repeatedly cycling through the elements of $l$ in reverse order.
The first element of $l$ becomes the last in the new list, the second becomes the second-to-last, and so on, until the list reaches $n$ elements.\footnote{This function is called \texttt{repeat\_seq} in our Rocq code.}
For example, $\mathit{repeat\_pattern} ([1;0],3) = [1;0;1]$ and $\mathit{repeat\_pattern} ([1;0],4) = [0;1;0;1]$.

\begin{prop}[Output Behaviors with Input Containing One $1$ in a Positive Loop]~\texttt{[PL\_Oscillation\_1\_true]}\\\ \\
$\begin{array}{l}
  \forall (\mathit{NC}: \mathit{NeuroCircuit}) (\mathit{inp}_1 \mathit{inp}_2:\mathit{list}~\mathit{bool}),\\
  \quad \mathit{PositiveLoop}(\mathit{NC}) ~ \land  ~
  \mathit{is\_initial}_\mathit{Cir}(\mathit{NC}) ~\land ~
  (\forall (b: \mathit{bool}), b \in \mathit{inp}_1 ~\mathop{++}~ \mathit{inp}_2 \rightarrow b = 0)~\land \\
  \quad w_{N_{0}}(2) \geq \tau_{N_{0}}~\land ~ w_{N_{0}}(1) \geq \tau_{N_{0}}~\land ~ w_{N_{1}}(0) \geq \tau_{N_{1}}~\rightarrow \\
  \quad\quad \mathit{output}_\mathit{NC}(N_{0}, \mathit{inp}_1~\mathop{++}~[1]~\mathop{++}~\mathit{inp}_2) = \\
  \quad\quad\quad\quad \mathit{repeat\_pattern}([1;0], \mathit{length} (\mathit{inp}_1)+1)~\mathop{++}~\mathit{repeat}(0, \mathit{length} (\mathit{inp}_2)+1) ~ \land ~\\
  \quad\quad\mathit{output}_\mathit{NC}(N_{1}, \mathit{inp}_1~\mathop{++}~[1]~\mathop{++}~\mathit{inp}_2) = \\
  \quad\quad\quad\quad \mathit{repeat\_pattern}([1;0], \mathit{length} (\mathit{inp}_1))~\mathop{++}~\mathit{repeat}(0, \mathit{length} (\mathit{inp}_2)+2).
  \end{array}$
\label{prop:posloop3}
\end{prop}
\begin{proof}
The proof proceeds by induction on the structure of the input sequence $\mathit{inp}_1$.

\noindent
\textbf{Base case}: $\mathit{inp}_1 = [\ ]$.

\noindent
\textbf{Neuron $N_0$}: By Proposition~\ref{lem:posloop1}, $\mathit{output}_\mathit{NC}(N_{0}, \mathit{inp}_2) = repeat(0, \mathit{length} (\mathit{inp}_2)+1)$.\\
We remark that if the external input at a time unit is $1$, the current potential of $N_0$ reaches its threshold and $N_0$ fires since $w_{N_{0}}(2) \geq \tau_{N_{0}}$ and $w_{N_{0}}(1) \geq 0$. Thus, we have:

$\begin{array}{l}
\mathit{output}_\mathit{NC}(N_{0}, [1] \mathop{++}\mathit{inp}_2) = [1] \mathop{++} repeat(0, \mathit{length} (\mathit{inp}_2)+1).
\end{array}$

\noindent
\textbf{Neuron $N_1$}: By the delayer effect theorem (Corollary~\ref{cor:delay}):

$\begin{array}{l}
\mathit{output}_\mathit{NC}(N_{1}, [1] \mathop{++}\mathit{inp}_2) = repeat(0, \mathit{length} (\mathit{inp}_2)+2).
\end{array}$

\noindent
\textbf{Induction case}: We assume that the property holds for $\mathit{inp_1}$, and we must show that it also holds for $\mathit{i::inp_1}$, where $i$ is an additional input value.
By the induction hypothesis:
$\begin{array}{l}
\mathit{output}_\mathit{NC}(N_{0}, \mathit{inp}_1 ~\mathop{++}~[1]~\mathop{++}~\mathit{inp}_2) = {}\\
\quad\mathit{repeat\_pattern}([1;0], \mathit{length} (\mathit{inp}_1)+1)~\mathop{++}~\mathit{repeat}(0, \mathit{length} (\mathit{inp}_2)+1) \mbox{ and } \\\mathit{output}_\mathit{NC}(N_{1}, \mathit{inp}_1 ~\mathop{++}~[1]~\mathop{++}~\mathit{inp}_2) = {} \\
\quad\mathit{repeat\_pattern}([1;0], \mathit{length} (\mathit{inp}_1))~\mathop{++}~\mathit{repeat}(0, \mathit{length} (\mathit{inp}_2)+2).
\end{array}$

\noindent
\textbf{Neuron $N_0$}:

\noindent
\textbf{Sub-case 1}:
\textit{the last output of $N_0$ is $1$ for the input sequence $\mathit{inp}_1\mathop{++}[1]\mathop{++}\mathit{inp}_2$}.\\
In this case, the last output of $N_1$ is $0$ by the induction hypothesis. Since both $N_1$ and the external input source provide $0$ as input and the fact that $N_0$ fires at the previous time unit, the current potential is $0$. Thus, $N_0$ does not fire and the new output is $0$.

\noindent\textbf{Sub-case 2}:
\textit{the last output of $N_0$ is $0$ for the input sequence $\mathit{inp}_1\mathop{++}[1]\mathop{++}\mathit{inp}_2$}.\\
The last output of $N_1$ is $1$. Thus, since $w_{N_{0}}(1) \geq \tau_{N_{0}}$ and $w_{N_{0}}(2) \geq 0$, the current potential reaches the threshold, $N_0$ fires and the new output is $1$.
We deduce that:

$\begin{array}{l}
\mathit{output}_\mathit{NC}(N_{0}, i::\mathit{inp}_1 ~\mathop{++}~[1]~\mathop{++}~\mathit{inp}_2) = {} \\
\quad\mathit{repeat\_pattern}([1;0], \mathit{length} (i::\mathit{inp}_1)+1)~\mathop{++}~\mathit{repeat}(0, \mathit{length} (\mathit{inp}_2)+1).
\end{array}$

\noindent
\textbf{Neuron $N_1$}: By the delayer effect theorem (Corollary~\ref{cor:delay}):

$\begin{array}{l}
\mathit{output}_\mathit{NC}(N_{1}, i::\mathit{inp}_1 ~\mathop{++}~[1]~\mathop{++}~\mathit{inp}_2) = {}\\
\quad\mathit{repeat\_pattern}([1;0], \mathit{length} (i::\mathit{inp}_1))~\mathop{++}~\mathit{repeat}(0, \mathit{length} (\mathit{inp}_2)+2).
\end{array}$
\end{proof}
\subsection{Negative Loop}
\label{subsec:negloop}
The properties discussed in this section pertain to the \texttt{NegativeLoop} archetype. 
A negative loop prevents repetitive firing.
This archetype combines the interplay of positive and negative influences.
The properties we consider here assume a consistent pattern of only $1$s as input, and examine under what conditions the output will exhibit
a repeating pattern of two zeros followed by two ones.

To initiate firing, the weight of $N_0$ for the external input source, $w_{N_0}(2)$, must reach the threshold; otherwise, the circuit remains inactive. Similarly, the weight $w_{N_1}(0)$ must reach the threshold for $N_1$ to fire at time 2, as $N_1$ only receives input from $N_0$, and it takes at least 2 time units for $N_1$ to activate.
We provide two distinct hypotheses under which these properties hold.

The first hypothesis involves a simple constraint on the weights, which appears on the fourth line in the statement below. Specifically, we assume that $w_{N_0}(1)$ is the opposite of $w_{N_0}(2)$. This constraint cancels out the positive potential created by the external input source when $N_1$ fires at a previous time step.  

\begin{prop}[Output Behaviors with Input Containing $1$s in a Negative Loop (a)]~ \texttt{[NL\_Output\_Oscil\_case1]}\\\ \\
$\begin{array}{l}
  \forall (\mathit{NC}: \mathit{NeuroCircuit}) (\mathit{inp}:\mathit{list}~\mathit{bool}),\\
  \quad \mathit{NegativeLoop}(\mathit{NC}) ~ \land  ~
  \mathit{is\_initial}_\mathit{Cir}(\mathit{NC}) ~\land ~ (\forall (b: \mathit{bool}), b \in \mathit{inp} \rightarrow b = 1)~\land \\
  \quad w_{N_{0}}(2) \geq \tau_{N_{0}}~\land ~ w_{N_{1}}(0) \geq \tau_{N_{1}} ~\land {}\\
  \quad w_{N_{0}}(1) = - w_{N_{0}}(2)~\rightarrow \\
  \quad\quad \mathit{output}_\mathit{NC}(N_{0}, \mathit{inp}) = \mathit{repeat\_pattern}([0;1;1;0], \mathit{length} (\mathit{inp})+1) ~ \land ~\\
  \quad\quad\mathit{output}_\mathit{NC}(N_{1}, \mathit{inp}) = \mathit{repeat\_pattern}([0;0;1;1], \mathit{length} (\mathit{inp})+1).
  \end{array}$
\label{prop:negloop1}
\end{prop}

\begin{proof}
The neuron with identifier $1$ is an initialized single-input neuron, with the weight of its sole input source meeting the threshold $\tau_{N_1}$. These properties fulfill the requirements to apply the filtering effect theorem.
The proof proceeds by induction on the structure of the input sequence $\mathit{inp}$.

\noindent
\textbf{Base case}: $\mathit{inp} = [\ ]$.\\ 
The output lists of both neuron are $[0]$ by design of our model.

\noindent
\textbf{Induction case}: We assume that the property holds for $\mathit{inp}$, and we must show that it also holds for $\mathit{i::inp}$, where $i$ is an additional input value.
By the induction hypothesis:

$\begin{array}{l}
\mathit{output}_\mathit{NC}(N_{0}, \mathit{inp}) = \mathit{repeat\_pattern}([0;1;1;0], \mathit{length} (\mathit{inp})+1) \mbox{ and }\\
\mathit{output}_\mathit{NC}(N_{1}, \mathit{inp}) = \mathit{repeat\_pattern}([0;0;1;1], \mathit{length} (\mathit{inp})+1).
\end{array}$

\noindent
\textbf{Neuron $N_0$}:
We remark that the current potential of $N_0$ is always non-negative as the external input provides $1$ at any time and any negative weight is compensated by the weight $w_{N_0}(2)$.
By the design of the circuit and the input sequence, $N_0$ fires only if the input from the external source is $1$ and from $N_1$ is $0$ at the previous time step.
To verify this statement, we only need to consider whether $N_1$ fired or not in the previous step, as the external input is always $1$.
If both $N_1$ and the external input source provide a $1$ as input, the output of $N_0$ at the previous step is $1$ by the induction hypothesis, and thus $N_0$ has fired at the previous step. The current potential of $N_0$ is equivalent to $0$ since $ w_{N_{0}}(1) = - w_{N_{0}}(2)$ and $N_0$ does not fire at the current time.
If $N_1$ does not fire at the previous step, the current potential of $N_0$ reaches the threshold $\tau_{N_0}$ since the weight $w_{N_0}(2)$ is greater than $\tau_{N_0}$ and any residual potential from previous steps are greater than or equal at 0. As a result, $N_0$ fires.

\noindent
\textbf{Sub-case 1:} \textit{$\mathit{length}(\mathit{inp})$ modulo $4$ is $0$ or $1$}:

\noindent
The output of $N_1$ is $0$ at time $\mathit{length}(\mathit{inp})$. (At that time, the input sequence processed is $\mathit{inp}$.) Thus, $N_0$ fires at time $\mathit{length}(\mathit{inp})+1$. Thus the property is verified.

\noindent
\textbf{Sub-case 2:} \textit{$\mathit{length}(\mathit{inp})$ modulo $4$ is $2$ or $3$:}

\noindent
The output of $N_1$ is $1$ at time $\mathit{length}(\mathit{inp})$. $N_0$ does not fire at time $\mathit{length}(\mathit{inp})+1$ and the property is verified.

\noindent
\textbf{Neuron $N_1$}: By the delayer effect theorem (Corollary~\ref{cor:delay}):

$\begin{array}{l}
\mathit{output}_\mathit{NC}(N_{1}, i::\mathit{inp}) = \mathit{repeat\_pattern}([0;1;1;0], \mathit{length} (\mathit{inp})+2)[1:\ ]:: [0] = {} \\
\quad\mathit{repeat\_pattern}([0;1;1;0], \mathit{length} (\mathit{inp})+1):: [0]= {}\\
\quad\quad\mathit{repeat\_pattern}([0;0;1;1], \mathit{length} (\mathit{inp})+2)
\end{array}$
\end{proof}

The second property is based on more complex hypotheses regarding the weights. To ensure that $N_0$ does not fire during two successive time steps when $N_1$ fires after having been inactive the previous step, we impose the condition:
$(1 + \mathit{lk}_{N_0})\cdot(w_{N_{0}}(1) + w_{N_{0}}(2)) < \tau_{N_0}$.
The left-hand side of the strict inequality corresponds to the current potential of neuron $N_0$ in the case where it did not fire in the previous time step, but did fire in the one before that.
In addition, $N_1$ must have fired at both time steps.
This condition prevents $N_0$ from firing during the specified time step. Additionally, it also prevents firing if both $N_0$ and $N_1$ fired in the previous step, since $1 \leq 1 + \mathit{lk}_{N_0}$.

We further add the hypothesis $w_{N_0}(1) + w_{N_0}(2) \geq 0$. This ensures that the current potential is never negative. As a result, $N_0$ does not need to compensate for a negative potential to fire.  Both hypotheses appear on the fourth line in the theorem statement below.

\begin{prop}[Output Behaviors with Input Containing $1$s in a Negative Loop (b)]~ \texttt{[NL\_Output\_Oscil\_case2]}\\\ \\
$\begin{array}{l}
  \forall (\mathit{NC}: \mathit{NeuroCircuit}) (\mathit{inp}:\mathit{list}~\mathit{bool}),\\
  \quad \mathit{NegativeLoop}(\mathit{NC}) ~ \land ~ 
  \mathit{is\_initial}_\mathit{Cir}(\mathit{NC}) ~\land ~ (\forall (b: \mathit{bool}), b \in \mathit{inp} \rightarrow b = 1)~\land \\
  \quad  w_{N_{0}}(2) \geq \tau_{N_{0}}~ \land ~ w_{N_{1}}(0) \geq \tau_{N_{1}}~ \land {}\\
  \quad  w_{N_{0}}(1) + w_{N_{0}}(2) \geq 0 ~\land~
  (1 + \mathit{lk}_{N_0})\cdot(w_{N_{0}}(1) + w_{N_{0}}(2)) < \tau_{N_0} ~\rightarrow \\
  \quad\quad \mathit{output}_\mathit{NC}(N_{0}, \mathit{inp}) = \mathit{repeat\_pattern}([0;1;1;0], \mathit{length} (\mathit{inp})+1) ~ \land ~\\
  \quad\quad\mathit{output}_\mathit{NC}(N_{1}, \mathit{inp}) = \mathit{repeat\_pattern}([0;0;1;1], \mathit{length} (\mathit{inp})+1).
  \end{array}$
\label{prop:negloop2}
\end{prop}

\begin{proof}
The neuron with identifier $1$ is an initialized single-input neuron, with the weight of its sole input source meeting the threshold $\tau_{N_1}$. These properties fulfill the requirements to apply the filtering effect theorem.

We note that since $w_{N_{0}}(1) + w_{N_{0}}(2) \geq 0 $ and that the external source of inputs only contains $1$s, the current potential is never negative.
The proof proceeds by induction on the structure of the input sequence $\mathit{inp}$.

\noindent
\textbf{Base case}: $\mathit{inp} = [\ ]$.

\noindent
The output list of both neurons is $[0]$ by design of our model.

\noindent
\textbf{Induction case}: We assume that the property holds for $\mathit{inp}$, and we must show that it also holds for $\mathit{i::inp}$, where $i$ is an additional input value.
By the induction hypothesis: 

$\begin{array}{l}
\mathit{output}_\mathit{NC}(N_{0}, \mathit{inp}) = \mathit{repeat\_pattern}([0;1;1;0], \mathit{length} (\mathit{inp})+1) \mbox{ and } \\
\mathit{output}_\mathit{NC}(N_{1}, \mathit{inp}) = \mathit{repeat}\_pattern([0;0;1;1], \mathit{length} (\mathit{inp})+1).
\end{array}$

\noindent
\textbf{Neuron $N_0$}:

\noindent
\textbf{Sub-case 1:} \textit{$\mathit{length}(\mathit{inp})$ modulo $4$ is $0$ or $1$}:

\noindent
The output of $N_1$ is $0$ at time $\mathit{length}(\mathit{inp})$. Since the current potential is never negative and $w_{N_{0}}(2) \geq \tau_{N_0}$, the threshold is reached and $N_0$ fires at time $\mathit{length}(\mathit{inp})+1$. The property is verified.

\noindent
\textbf{Sub-case 2:} \textit{$\mathit{length}(\mathit{inp})$ modulo $4$ is $2$}:

\noindent
The output of $N_0$ is $1$ and of $N_1$ is $1$ at time $\mathit{length}(\mathit{inp})$. Thus, $\mathit{curpot}_\mathit{NC}(N_0,i::\mathit{inp})\equiv w_{N_0}(1) + w_{N_0}(2)$. Since $(1 + \mathit{lk}_{N_0})\cdot(w_{N_{0}}(1) + w_{N_{0}}(2)) < \tau_{N_0}$ and $(1 + \mathit{lk}_{N_0}) \geq 1$, we deduce that:  $\mathit{curpot}_\mathit{NC}(N_0,i::\mathit{inp}) < \tau_{N_0}$. We have that $N_0$ does not fire and the property is verified.

\noindent
\textbf{Sub-case 3:} \textit{$\mathit{length}(\mathit{inp})$ modulo $4$ is $3$}:

\noindent
The output of $N_0$ is $0$ and of $N_1$ is $1$ at time $\mathit{length}(\mathit{inp})$. Since $\mathit{length}(\mathit{inp})$ modulo $4$ is $3$ implies that $\mathit{length}(\mathit{inp})\geq 3$, the output of $N_0$ is $1$ and of $N_1$ is $1$ at time $\mathit{length}(\mathit{inp}) - 1$. By the same reasoning as Sub-case 2, $\mathit{curpot}_\mathit{NC}(N_0,\mathit{inp})\equiv w_{N_0}(1) + w_{N_0}(2)$ and thus, $\mathit{curpot}_\mathit{NC}(N_0,i::\mathit{inp})\equiv (1 + lk_{N_0}) \cdot (w_{N_0}(1) + w_{N_0}(2))$. By the second hypothesis on the fourth line of the theorem statement,  $\mathit{curpot}_\mathit{NC}(N_0,i::\mathit{inp}) < \tau_{N_0}$. $N_0$ does not fire and the property is verified.

\noindent
\textbf{Neuron $N_1$}: By the delayer effect theorem (Corollary~\ref{cor:delay}):

$\begin{array}{l}
\mathit{output}_\mathit{NC}(N_{1}, i::\mathit{inp}) = \mathit{repeat\_pattern}([0;1;1;0], \mathit{length} (\mathit{inp})+2)[1:\ ]:: [0]={} \\
\quad \mathit{repeat\_pattern}([0;1;1;0], \mathit{length} (\mathit{inp})+1):: [0]={}\\ \quad\quad\mathit{repeat\_pattern}([0;0;1;1], \mathit{length} (\mathit{inp})+2).
\end{array}$
\end{proof}

\subsection{Contralateral Inhibition}
\label{subsec:continhib}
This archetype, also known as mutual inhibition,  plays an important role in several behavioural mechanisms of different living creatures. For instance, active and passive fear responses are mediated by distinct and mutually inhibitory central amygdala neurons \cite{FK2017Nature}.
For this archetype, the property of interest is known as \emph{winner takes all}. Under specific hypotheses, at a certain time step, one neuron dominates by continuously firing, thereby preventing the other neuron from firing. In this scenario, we assume that the external inputs are fixed at $1$.
This property holds when:
\begin{itemize}
\item the sum of the weights of the first neuron is greater than the threshold---$N_0$ fires under all circumstances;
\item the weight of the external inputs to the second neuron is greater than the threshold---$N_1$ fires when $N_0$ is inactive;
\item but the sum of the weights of the second neuron is negative---$N_1$'s firing is blocked when $N_0$ is active.
\end{itemize}
\begin{prop}[Output Behaviors in Contralateral Inhibition] \texttt{[CI\_Winner\_Takes\_All]}\\\ \\
$\begin{array}{l}
  \forall (\mathit{NC}: \mathit{NeuroCircuit}) (\mathit{inp}:\mathit{list}~(\mathit{nat} \rightarrow\mathit{bool})),\\
  \quad \mathit{ContraInhib}(\mathit{NC}) ~ \land ~ 
  \mathit{is\_initial}_\mathit{Cir}(\mathit{NC}) ~\land ~ \\
  \quad(\forall (b: \mathit{bool}), b \in \mathit{inp} \rightarrow b(2)= 1~\land~ b(3)= 1)~\land {} \\
  \quad w_{N_{0}}(1) + w_{N_{0}}(2) \geq \tau_{N_0}~ \land ~ w_{N_{1}}(3) \geq \tau_{N_{1}}~\land ~ w_{N_{1}}(0) + w_{N_{1}}(3) \leq 0~\rightarrow \\
  \quad\quad \mathit{output}_\mathit{NC}(N_{0}, \mathit{inp}) = \mathit{repeat}(1, \mathit{length} (\mathit{inp}))~\mathop{++}~[0] ~ \land ~\\
  \quad\quad\mathit{output}_\mathit{NC}(N_{1}, \mathit{inp}) =\begin{cases}
  [0]&\mathit{if}~\mathit{len} (\mathit{inp}) = 0 \\
  \mathit{repeat}(0, \mathit{length} (\mathit{inp})-1)\mathop{++} [1; 0] & \mathit{otherwise}
  \end{cases}
  \end{array}$
\label{prop:contra}
\end{prop}
\begin{proof}
Recall that by the definition of $\mathit{ContraInhib}(\mathit{NC})$, the identifier of the external source of input to $N_0$ is $2$, and to $N_1$ is $3$.  
Also, $w_{N_0}(1)$ is negative.

\noindent
\textbf{Neuron $N_0$}: 

\noindent
We note that $w_{N_{0}}(1) + w_{N_{0}}(2) \geq \tau_{N_0}$ implies 
$w_{N_{0}}(2) \geq \tau_{N_0}$ since $w_{N_{0}}(1)$ is negative. 
As a result, if the external source of input provides $1$ as input to $N_0$, $N_{0}$ fires. 
By hypothesis, the external inputs are always $1$, so as soon as the circuit recieves an input, $N_0$ fires. Since the output list at time $0$ is $[0]$, we have: $ \mathit{output}_\mathit{NC}(N_{0}, \mathit{inp}) = \mathit{repeat}(1, \mathit{length} (\mathit{inp}))~\mathop{++}~[0]$.

\noindent
\textbf{Neuron $N_1$}:
The output list of $N_1$ at time $0$ is $[0]$ by construction of our model.

At time $1$, the last output of $N_0$ is $0$ and the external source of input provides $1$ as input. Since $w_{N_{1}}(3) \geq \tau_{N_{1}}$, $N_1$ fires.

For any time $n\geq 2$, the last output of $N_0$ is $1$ and the external source of input provides $1$ as input. By the hypothesis $w_{N_{1}}(0) + w_{N_{1}}(3) \leq 0$, we deduce that the current potential is always non-positive starting at time $2$. $N_1$ does not fire at time $n$.
Thus, we have:

$\mathit{output}_\mathit{NC}(N_{1}, \mathit{inp}) =\begin{cases}
  [0]&\mathit{if}~\mathit{len} (\mathit{inp}) = 0 \\
  \mathit{repeat}(0, \mathit{length} (\mathit{inp})-1)\mathop{++} [1; 0] & \mathit{otherwise}
  \end{cases}$
\end{proof}

\section{Conclusion}
\label{sec:concl}

In this work, we proposed a formal approach to model and validate Leaky Integrate-and-Fire neurons and archetypes.
In the literature, this is not the first attempt to the formal investigation of neural networks.
In~\cite{DMGRG16HSB,DLGRG17CSBIO}, the synchronous paradigm has been exploited to model neurons and some small neuronal circuits with a relevant topological structure and behavior and to prove some properties concerning their dynamics.
Our approach based on the use of the Rocq Prover (which is, to the best of our knowledge, the first one) turned out to be much more general.
As a matter of fact, we guarantee that the properties we prove are true in the general case, such as true for any input values, any length of input, and any amount of time.
As an example, let us consider the simple series.
In~\cite{DLGRG17CSBIO}, the authors were able to write a function (more precisely, a Lustre node) which encodes the expected behavior of the circuit.
Then, they could call a model checker to test whether the property at issue is valid for some input series with a fixed length.
Here we can prove that the desired behavior is true no matter what the length is and what the parameters of the series are.

As seen in the definition of the archetypes, the neurons are numbered in a specific order.
However, we want to ensure that a circuit is considered a particular archetype even if the identifiers of the neurons and external sources do not follow this specific numbering order.
To address this, we are working on defining an equivalence relation that verifies whether two circuits are equivalent by checking if there exists a rotation such that, after applying it to one of the circuits, both circuits become identical.
This equivalence relation will allow for more flexible classification of circuits, ensuring that the same archetype can be recognized regardless of the labeling or ordering of the neurons and external inputs.

One of our long-term goals is to verify that every neuronal circuit in nature is a composition of multiple archetypes.
A first step toward this goal is the definition of a method for composing two circuits together.
One approach is to compose circuits sequentially.
In this case, a neuron $N_1$ from the first circuit provides its output as input to a neuron $N_2$ in the second circuit, with $N_1$ replacing one of the external sources of $N_2$.
Another approach involves directly integrating two circuits.
In this case, a circuit $\mathit{NC}_1$ is plugged into another circuit $\mathit{NC}_2$, where one of the neurons in $\mathit{NC}_2$ is removed and replaced by the entire circuit $\mathit{NC}_1$.
The equivalence relation and circuit composition methods discussed are part of ongoing work and will be presented in a future paper.
This work builds directly on our library of lemmas and definitions used in this paper.

\bibliographystyle{alphaurl}
\bibliography{neuronal}

\appendix

\section{Rocq Definitions of Archetypes}
\label{app:archetypes}

Figures~\ref{fig:CoqNegloop},~\ref{fig:CoqInhib}, and~\ref{fig:CoqContra} contain the full definitions of the negative loop, inhibition, and contralateral inhibition archetypes, respectively, which were discussed at the end of Section~\ref{subsec:archetypes-props}.
\begin{figure}[htb]
\begin{verbatim}
Record NegativeLoop (c : NeuroCircuit) :=
  Make_NegLoop
    {
      OneSupplementNL : SupplInput c = 1;
      ListNeuroLengthNL : length (ListNeuro c) = 2;
      FirstNeuronNL : forall n,
        In n (ListNeuro c) -> Id (Feature n) = 0 ->
        Weights (Feature n) 1 < 0 /\ 0 < Weights (Feature n) 2;
      SecondNeuroNL : forall n,
        In n (ListNeuro c) -> Id (Feature n) = 1 ->
        0 < Weights (Feature n) 0 /\ Weights (Feature n) 2 == 0;
    }.
\end{verbatim}
\caption{Rocq representation of the negative loop archetype}
\label{fig:CoqNegloop}  
\end{figure}

\begin{figure}[htb]
\begin{verbatim}
Record Inhibition (c : NeuroCircuit) :=
  Make_Inhib
    {
      OneSupplementI : SupplInput c = 2;
      ListNeuroLengthI : length (ListNeuro c) = 2;
      FirstNeuronI : forall n,
        In n (ListNeuro c) -> Id (Feature n) = 0 ->
        Weights (Feature n) 1 == 0 /\ 0 < Weights (Feature n) 2 /\
        Weights (Feature n) 3 == 0;
      SecondNeuroI : forall n,
        In n (ListNeuro c) -> Id (Feature n) = 1 ->
        Weights (Feature n) 0 < 0 /\ Weights (Feature n) 2 == 0 /\
        0 < Weights (Feature n) 3;
    }.
\end{verbatim}
\caption{Rocq representation of the inhibition archetype}
\label{fig:CoqInhib}  
\end{figure}

\begin{figure}[htb]
\begin{verbatim}
Record ContraInhib (c : NeuroCircuit) :=
  Make_ContrInhib
    {
      OneSupplementCI : SupplInput c = 2;
      ListNeuroLengthCI : length (ListNeuro c) = 2;
      FirstNeuronCI : forall n,
        In n (ListNeuro c) -> Id (Feature n) = 0 ->
        Weights (Feature n) 1 < 0 /\ 0 < Weights (Feature n) 2 /\
        Weights (Feature n) 3 == 0;
      SecondNeuroCI : forall n,
        In n (ListNeuro c) -> Id (Feature n) = 1 ->
        Weights (Feature n) 0 < 0 /\ Weights (Feature n) 2 == 0 /\
        0 < Weights (Feature n) 3;
    }.
\end{verbatim}
\caption{Rocq representation of the contralateral inhibition archetype}
\label{fig:CoqContra}  
\end{figure}

\end{document}